\documentclass[11pt, letterpaper]{article}

\usepackage{latexsym}
\usepackage{amsmath, amsfonts, amssymb, amsthm}
\usepackage{bbm}

\usepackage{appendix}

\usepackage{enumerate}
\usepackage{algorithm, algpseudocode}
\usepackage{tikz}
\usepackage{hyperref}
\usepackage[small,compact]{titlesec}
\usepackage{times}

\usepackage{paralist}

\usepackage{bbm}

\usepackage{url}

\usepackage{tabu}

\usepackage{xfrac}

\usepackage[normalem]{ulem}

\usepackage{nicematrix}

\hypersetup{
    colorlinks=true,
    pdftitle={Overleaf Example},
    pdfpagemode=FullScreen,
}

\newtheorem{theorem}{Theorem}
\newtheorem{corollary}[theorem]{Corollary}
\newtheorem{lemma}[theorem]{Lemma}

\newtheorem{proposition}[theorem]{Proposition}
\newtheorem{definition}[theorem]{Definition}
\newtheorem{claim}[theorem]{Claim}

\theoremstyle{remark}
\newtheorem{remark}[theorem]{Remark}

\newcommand{\rank}{\mathrm{rank}}
\newcommand{\nnz}{\mathrm{nnz}}
\newcommand{\size}{\mathrm{size}}
 
\newcommand{\bo}{\boldsymbol 1}
\newcommand{\rig}{\mathcal R}

\newcommand{\ceil}[1]{\left\lceil #1 \right\rceil}
\newcommand{\floor}[1]{\left\lfloor #1 \right\rfloor}
\newcommand{\ip}[2]{\langle #1, #2 \rangle}

\newcommand{\C}{\ensuremath{\mathbb{C}}}

\newcommand{\F}{\ensuremath{\mathbb{F}}}
\newcommand{\K}{\ensuremath{\mathbb{K}}}
\newcommand{\R}{\ensuremath{\mathbb{R}}}
\newcommand{\Q}{\ensuremath{\mathbb{Q}}}
\newcommand{\Z}{\ensuremath{\mathbb{Z}}}

\newcommand{\CF}{\ensuremath{\mathcal{F}}}
\newcommand{\CG}{\ensuremath{\mathcal{G}}}
\newcommand{\CH}{\ensuremath{\mathcal{H}}}
\newcommand{\CU}{\ensuremath{\mathcal{U}}}
\newcommand{\CV}{\ensuremath{\mathcal{V}}}
\newcommand{\CC}{\ensuremath{\mathcal{C}}}
\newcommand{\CD}{\ensuremath{\mathcal{D}}}
\newcommand{\CT}{\ensuremath{\mathcal{T}}}
\newcommand{\CE}{\ensuremath{\mathcal{E}}}

\newcommand{\TX}{\ensuremath{\widetilde{X}}}
\newcommand{\TY}{\ensuremath{\widetilde{Y}}}

\newcommand{\TE}{\ensuremath{\widetilde{E}}}
\newcommand{\TF}{\ensuremath{\widetilde{F}}}

\newcommand{\MS}{\ensuremath{\mathrm{MS}}}

\newcommand{\eps}{\varepsilon}
\newcommand{\Ex}{\mathbb{E}}

\numberwithin{equation}{section}
\numberwithin{theorem}{section}

\advance \topmargin by -\headheight
\advance \topmargin by -\headsep
\evensidemargin \oddsidemargin
\title{Smaller Low-Depth Circuits for Kronecker Powers}
\author{Josh Alman\thanks{Department of Computer Science at Columbia University. josh@cs.columbia.edu.}
\and Yunfeng Guan\thanks{Department of Computer Science at Columbia University. yunfeng.guan@columbia.edu.}
\and Ashwin Padaki\thanks{Department of Computer Science at Columbia University. ap4025@columbia.edu.}}

\begin{document}

\maketitle

\begin{abstract}
A linear circuit for computing an $N \times N$ matrix $M$ is a circuit with $N$ inputs corresponding to the entries of a vector $x$ and $N$ outputs corresponding to the entries of the transformed vector $Mx$, and where each gate computes a linear combination of its inputs. Each gate may have unbounded fan-in, and the size of the circuit is the number of wires. This model captures most known algorithms for computing linear transforms, and (in the constant-depth or `synchronous' settings) is equivalent to factoring $M$ as the product of sparse matrices.

We give new, smaller constructions of constant-depth linear circuits for computing any matrix which is the Kronecker power of a fixed matrix. A standard argument (e.g., the mixed product property of Kronecker products, or a generalization of the Fast Walsh-Hadamard transform) shows that any such $N \times N$ matrix has a depth-2 circuit of size $O(N^{1.5})$. We improve on this for all such matrices, and especially for some such matrices of particular interest:
    \begin{itemize}
        \item For any integer $q > 1$ and any matrix which is the Kronecker power of a fixed $q \times q$ matrix, we construct a depth-2 circuit of size $O(N^{1.5 - a_q})$, where $a_q > 0$ is a positive constant depending only on $q$. No bound beating size $O(N^{1.5})$ was previously known for any $q>2$.
        \item For the case $q=2$, i.e., for any matrix which is the Kronecker power of a fixed $2 \times 2$ matrix, we construct a depth-2 circuit of size $O(N^{1.446})$, improving the prior best size $O(N^{1.493})$ [Alman, 2021].
        \item For the Walsh-Hadamard transform, we construct a depth-2 circuit of size $O(N^{1.443})$, improving the prior best size $O(N^{1.476})$ [Alman, 2021].
        \item For the disjointness matrix (the communication matrix of set disjointness, or equivalently, the matrix for the linear transform that evaluates a multilinear polynomial on all $0/1$ inputs), we construct a depth-2 circuit of size $O(N^{1.258})$, improving the prior best size $O(N^{1.272})$ [Jukna and Sergeev, 2013].
    \end{itemize}
    
Our constructions also generalize to improving the standard construction for any depth $\leq O(\log N)$. Our main technical tool is an improved way to convert a nontrivial circuit for any matrix into a circuit for its Kronecker powers. Our new bounds provably could not be achieved using the approaches of prior work.
\end{abstract}

\thispagestyle{empty}
\newpage
\setcounter{page}{1}
\section{Introduction}

\paragraph{Linear circuits}
In a \emph{linear circuit} over a field $\F$, each input is a number from $\F$, and each gate computes an $\F$-linear combination of its inputs. One can imagine labeling each edge of the circuit with a coefficient $\lambda \in \F$, so that if a gate $G$ has incoming edges from gates $G_1, \ldots, G_k$ labeled with $\lambda_1, \ldots, \lambda_k$, respectively, then $G$ computes $\lambda_1 G_1 + \cdots + \lambda_k G_k$. Hence, a linear circuit computes a linear transformation of its input. Such a circuit with $N_1$ outputs and $N_2$ inputs corresponds to a matrix $M \in \F^{N_1 \times N_2}$ such that, on input $x \in \F^{N_2}$, it outputs $Mx$. The size of the circuit is the total number of wires, and the depth is the length of the longest path from an input to an output.

The linear circuit model naturally captures essentially all known algorithms for computing linear transforms such as the Fast Walsh-Hadamard transform and the Fast Fourier transform (see e.g.,\cite{lokam2009complexity,burgisser2013algebraic}). For instance, as we will discuss more shortly, over any field, the $N \times N$ Walsh-Hadamard transform is computed by a linear circuit of size $O(N \log N)$ and depth $O(\log N)$, and for any constant $d$, a linear circuit of size $O(N^{1 + 1/d})$ and depth $d$ by using the Fast Walsh-Hadamard transform. Furthermore, any arithmetic circuit for computing a linear transform can be converted into a linear circuit without increasing the depth or size of the circuit~\cite[{Theorem~13.1}]{burgisser2013algebraic}.

A special kind of linear circuit is a \emph{synchronous} linear circuit, in which all paths from an input to an output have the same length. Equivalently, if we say that the depth of a gate $G$ is the length of the longest path from an input to $G$ in the circuit, then in a synchronous linear circuit, all the inputs to any gate must have the same depth. Any linear circuit of depth $d$ can be converted to an equivalent synchronous linear circuit by copying each gate fewer than $d$ times, and hence multiplying the size of the circuit by at most $d$. We especially focus in this paper on constant-depth circuits, where this is a negligible size increase.

A synchronous linear circuit for matrix $M \in \F^{N_1 \times N_2}$ can also equivalently be viewed as a factorization of $M$ as the product of sparse matrices. More precisely, $M$ has a synchronous linear circuit of size $s$ and depth $d$ if and only if it can be written as $M = M_d \times M_{d-1} \times \cdots \times M_1$, where the $M_i$ are matrices over $\F$ of any dimensions (such that the multiplications are all valid) and such that $\nnz(M_1) + \cdots + \nnz(M_d) = s$ (where $\nnz(M_i)$ denotes the number of nonzero entries in $M_i$). The matrix $M_i$ computes all the gates at depth $i$ in the circuit, and $\nnz(M_i)$ counts the number of wires used by those gates.
In particular, any matrix $M$ is computed by a circuit of depth $1$ and size $\nnz(M)$, and one typically aims to use more depth to design smaller circuits than this.

\paragraph{Kronecker products}
For two matrices $A \in \F^{N_1 \times N_2}$ and $B \in \F^{N_3 \times N_4}$, their \emph{Kronecker product} is the matrix $A \otimes B \in \F^{N_1 N_3 \times N_2 N_4}$ given by, for $i_1 \in [N_1]$, $i_2 \in [N_2]$, $i_3 \in [N_3]$ and $i_4 \in [N_4]$,
$$A \otimes B[i_1 + N_1 i_3, i_2 + N_2 i_4]  = A[i_1, i_2] \cdot B[i_3, i_4].$$ For a nonnegative integer $n$, the \emph{Kronecker power} $A^{\otimes n} = A \otimes \cdots \otimes A \in \F^{N_1^n \times N_2^n}$ is the Kronecker product of $n$ copies of $A$. For example, the Walsh-Hadamard transform $H_n \in \{-1,1\}^{2^n \times 2^n}$ is defined by $H_n = H_1^{\otimes n}$ where $$H_1 = \begin{pmatrix}
1 & 1\\
1 & -1
\end{pmatrix}.$$ Equivalently, $H_n$ is recursively-defined as $$H_n = \begin{pmatrix}
H_{n-1} & H_{n-1}\\
H_{n-1} & -H_{n-1}
\end{pmatrix}.$$

In this paper, we focus on linear circuits for computing matrices which can be expressed as a Kronecker power of a fixed matrix. $H_n$ is perhaps the most prominent example; it is used throughout algorithm design, complexity theory, signal processing, quantum computing (it is a basic quantum logic gate), and more (see e.g.,~\cite{macwilliams1977theory,beauchamp1984applications,wong2022walsh}). Another ubiquitous example is the \emph{disjointness matrix} $R_n = R_1^{\otimes n}$ where $$R_1 = \begin{pmatrix}
1 & 1\\
1 & 0
\end{pmatrix}.$$ $R_n$ is the communication matrix for the set disjointness function, and its arithmetic complexity is pertinent to the fastest known algorithm for the orthogonal vectors problem~\cite{abboud2014more} (see also~\cite{AW17}). More generally, Kronecker power matrices arise naturally in many settings. For example, for any $\omega \in \F$, $\omega \neq 1$, defining $$M_\omega = \begin{pmatrix}
1 & 1\\
1 & \omega
\end{pmatrix},$$ we see that multiplying $M_\omega^{\otimes n}$ times a vector $v \in \F^{2^n}$ is equivalent to evaluating the multilinear polynomial, whose coefficients are given by $v$, on all $2^n$ points in $\{1,\omega\}^n$. Again, $H_n$ and $R_n$ are prevalent examples where $\omega = -1$ and $\omega = 0$, respectively. Kronecker powers of larger matrices can similarly capture more general multivariate polynomial evaluation problems.

\paragraph{Mixed product property}

A powerful tool for designing linear circuits for Kronecker product matrices is the \emph{mixed product property} of the Kronecker product, which says that for matrices $A,B,C,D$ (for which the dimensions in the matrix products below work out), we have $$(A \times B) \otimes (C \times D) = (A \otimes C) \times (B \otimes D),$$ where $\times$ denotes the usual matrix product. It follows, for instance, that setting $N = 2^n$, the $N \times N$ matrix $H_n$ (which has $\nnz(H_n) = N^2$) has a depth-2 circuit of size only $O(N^{1.5})$ via $$H_n = H_{n/2} \otimes H_{n/2} = (H_{n/2} \times I_{\sqrt{N}}) \otimes (I_{\sqrt{N}} \times H_{n/2}) =  (H_{n/2} \otimes I_{\sqrt{N}}) \times (I_{\sqrt{N}} \otimes H_{n/2}),$$
where $I_{\sqrt{N}}$ denotes the $\sqrt{N} \times \sqrt{N}$ identity matrix. Indeed, this expresses $H_n$ as the product of two matrices, each of which has number of nonzero entries $$\nnz(H_{n/2} \otimes I_{\sqrt{N}}) = \nnz(H_{n/2}) \cdot \nnz(I_{\sqrt{N}}) = N \cdot \sqrt{N} = N^{1.5}.$$ Following the same idea, $H_n$ has a circuit of depth $d$ and size $O(d \cdot N^{1 + 1/d})$ (whenever $d$ is a constant or divides $n$), and setting $d = n$, we get the Fast Walsh-Hadamard transform, a circuit of depth $O(\log N)$ and size $O(N \log N)$. The same approach works for any Kronecker power matrix.

\paragraph{Barriers for prior approaches}
Until recently, this construction was believed to be optimal, and there were two popular conjectures toward proving this. First, $H_n$ was conjectured to have high `matrix rigidity'. The rank-$r$ rigidity of a matrix $M$, denoted $\mathcal{R}_M(r)$, is the minimum number of entries of $M$ one must change to make its rank at most $r$. (See subsection \ref{subsec:prelim-rig} for formal definition.) Valiant~\cite{valiant1977graph} showed that if the $N \times N$ matrix $M$ has $\mathcal{R}_M(N / \log \log N) \geq N^{1 + \eps}$, then $M$ doesn't have a circuit of depth $O(\log N)$ and size $O(N)$. However, Alman and Williams \cite{AW17} showed that the matrix rigidity of $H_n$ is surprisingly low, and that this approach cannot work to prove circuit lower bounds for $H_n$. Alman~\cite{Alm21} extended this result to show that no Kronecker power of a fixed-sized matrix is rigid enough to prove a lower bound in this way, and Kivva~\cite{kivva2021improved} proved even stronger rigidity upper bounds for these matrices. (The notion of matrix rigidity is very related to low-depth linear circuits; it will reappear in our results below.)

Second, Morgenstern~\cite{morgenstern1973note} and Chazelle~\cite{chazelle1994spectral} began a line of work on \emph{bounded-coefficient circuits} for computing $H_n$ and other important linear transforms (e.g., \cite{lokam2001spectral,nisan1996lower,pudlak2000note,burgisser2004lower,raz2002complexity}). They noticed that the Fast Walsh-Hadamard transform, Fast Fourier transform, and other similar algorithms all give linear circuits over $\C$ which have coefficients (the numbers $\lambda$ along the wires) whose magnitudes are bounded by a constant. For instance, the coefficients in the Fast Walsh-Hadamard transform are all from $\{1, -1\}$, and so have magnitude $1$. It was natural to assume that the best circuits for $H_n$, whose entries are all from $\{1, -1\}$, would also use bounded coefficients, and this line of work proved strong, often tight lower bounds on the sizes of such circuits. For example, Pudl{\'a}k~\cite{pudlak2000note} proved that bounded-coefficient circuits for $H_n$ of depth $2$ require size $\Omega(N^{1.5})$, meaning the Fast Walsh-Hadamard transform is tight in this depth-2, bounded-coefficient setting. Pudl{\'a}k also proved that synchronous linear circuits with bounded coefficients of any depth for $H_n$ require size $\Omega(N \log N)$. 

However, Alman~\cite{Alm21} recently used \emph{unbounded} coefficients to design smaller circuits for $H_n$ and other Kronecker products which circumvent these bounded-coefficient lower bounds, including a depth-2 circuit of size $O(N^{1.48})$, and an $O(N \log N)$ size circuit with a slightly improved constant factor compared to the Fast Walsh-Hadamard transform. Especially in light of the many applications of $H_n$ and other Kronecker power matrices, one naturally wonders how much smaller these circuits could get. 

We focus here especially on depth-2 circuits, since it is known that using smaller depth-2 circuits for Kronecker power matrices, one can directly construct smaller circuits of any depth $d \leq O(\log N)$:
\begin{lemma}[{Follows from the mixed product property; also implicit in \cite[{Theorem~3.4}]{Alm21}}]\label{lem:increasedepth}
Suppose that for some $c \geq 0$, the $q \times q$ matrix $M$ has a depth-2 synchronous circuit of size $q^{1 + c/2}$. Then, for any positive even integers $d, n$ such that $d$ divides $n$, the $N \times N$ matrix $M^{\otimes n}$ ($N = q^n$) has a depth-$d$ synchronous circuit of size $O(d \cdot N^{1 + c/d})$.
\end{lemma}
\noindent In fact, to our knowledge, all known nontrivial linear circuits for Kronecker powers with depth $d > 2$ could be achieved by giving a depth 2 construction and using Lemma~\ref{lem:increasedepth}. (A version of Lemma~\ref{lem:increasedepth} also holds for odd $d$ with a slight increase in the size.)

Unfortunately, Alman's approach, which makes use of the rank-1 rigidity of $H_n$, provably cannot give a smaller depth-2 circuit for $H_n$ than size $O(N^{1.47})$; we prove this in Subsection~\ref{subsec:comparison} below. In short, the approach of Alman makes use of circuit constructions (via rank-1 rigidity upper bounds) for a fixed matrix $M$ in order to design smaller circuits for all Kronecker powers $M^{\otimes n}$ of $M$. We are able to prove a rigidity lower bound, showing that there are not better constructions to use with their methods than the one they already use.
The only other approach we are aware of for designing small low-depth circuits for Kronecker products is an approach based on analyzing a simple branching tree process by Jukna and Sergeev~\cite{JS13}, which seems specific to sparse matrices like $R_n$ (which has $\nnz(R_n) = 3^n$). As we discuss more below, there are fundamental difficulties to using this approach to design smaller circuits for $R_n$ as well. 

\subsection{Results}

In this paper, we nonetheless design smaller circuits for $H_n$, $R_n$, and any Kronecker power of a fixed matrix. We get around the aforementioned obstacles by giving a \emph{new} approach for designing circuits for these matrices based on fixed constructions. At a high level, whereas Alman's approach converted any nontrivial circuit for $M$ into a nontrivial circuit for $M^{\otimes n}$, our new approach can only make use of sufficiently `imbalanced' circuits for $M$, but in exchange, it results in substantially smaller circuits for $M^{\otimes n}$.

\paragraph{Main Result}

For any constant $q$ and matrix $M \in \F^{q \times q}$, our main result converts a given small `imbalanced' depth-2 synchronous circuits for $M$ into small depth-2 synchronous circuits for Kronecker powers $M^{\otimes n}$. Suppose the depth-2 circuit for $M$ has $J$ gates in the middle layer. We can thus represent it as \begin{align}\label{eq:fixeddecomp}M = \sum_{j=1}^J U_j \times V_j^T,\end{align} for vectors $U_j, V_j \in \F^q$ which represent the inputs and outputs, respectively, of the $j$th gate. Assuming this decomposition is `imbalanced' (which we define in a moment), our main result gives:
\begin{theorem}\label{thm:mainintro}
If $M \in \F^{q \times q}$ is symmetric and has a decomposition (\ref{eq:fixeddecomp}) which is imbalanced (see Definition~\ref{def:imba} below), then the Kronecker power $M^{\otimes n} \in \F^{q^n \times q^n}$ has a depth-2 synchronous linear circuit of size $$\left( \sum_{j=1}^J \sqrt{\nnz(U_j) \cdot \nnz(V_j)} \right)^{n + o(n)}.$$
\end{theorem}
Theorem~\ref{thm:mainintro} also holds in certain cases when $M$ is not symmetric. One important case is when the decomposition (\ref{eq:fixeddecomp}) comes from a rigidity upper bound; these decompositions are `one-sided', a stronger notion than `imbalanced' meaning that $\nnz(U_j) \leq \nnz(V_j)$ for all $j$. See Remark~\ref{rem:asym} below for more details. That said, most matrices of interest like $H_n$ and $R_n$ are symmetric, and we need not worry about this detail.

Thus, as long as we can find a decomposition (\ref{eq:fixeddecomp}) of $M \in \F^{q \times q}$ (or a Kronecker power of $M$) with $$\sum_{j=1}^J \sqrt{\nnz(U_j) \cdot \nnz(V_j)} < q^{1.5},$$ we get a nontrivial depth-2 circuit for powers $M^{\otimes n}$. 
Note that it is straightforward to achieve $q^{1.5}$, for instance, by picking $J = q$, $U_j$ to be the $j$th column of $M$, and $V_j$ to be the vector with a $1$ in entry $j$, and a $0$ in all other entries. We will show soon how to beat $q^{1.5}$ for (Kronecker powers of) \emph{any} fixed matrix $M$, and beat it by quite a bit for matrices of interest like $H_n$ and $R_n$.
This will give smaller depth-2 circuits for these Kronecker power matrices and, by Lemma \ref{lem:increasedepth},
it will also give a smaller depth-$d$ $(d \ge 3)$ circuit.

\paragraph{Imbalanced decomposition}

We now define and give intuition for what it means for the decomposition (\ref{eq:fixeddecomp}) to be imbalanced. Define first the quantities
\[\alpha_1 = \ln (\sum_{j=1}^J \sqrt{\nnz(U_j) \cdot \nnz(V_j)}), \quad
\alpha_2 = \ln (\sqrt{\nnz(M) \cdot q}).\]
These measure the sizes of circuits for $M^{\otimes n}$: Our Theorem~\ref{thm:mainintro} gives a depth-2 circuit of size $\exp(\alpha_1)^{n + o(n)}$, whereas just using the mixed product property as described above gives a depth-2 circuit of size $\exp(\alpha_2)^{n + o(n)}$. Hence, $\alpha_2 - \alpha_1$ measures how much our new circuit improves on the mixed product property construction.

Next define the quantity
\[E = \frac{\sum_{j=1}^J  \ln (\nnz(U_j)/ \nnz(V_j)) \cdot \sqrt{\nnz(U_j) \cdot \nnz(V_j)}}{\sum_{j=1}^J \sqrt{\nnz(U_j) \cdot \nnz(V_j)}}.\] 
The logarithm, $\ln (\nnz(U_j)/ \nnz(V_j))$, measures whether the $j$th gate of the circuit for $M$ has more input or output wires: it becomes more positive as $\nnz(U_j)$ gets bigger than $\nnz(V_j)$, and more negative otherwise. In other words, $\ln (\nnz(U_j)/ \nnz(V_j))$ measures how `imbalanced' the $j$th gate is. Then, $E$ is a weighted average of these `imbalances', weighted by how much each gate contributes to $\exp(\alpha_1)$. It is thus a measure of the `imbalance' of the whole circuit. Note that by possibly taking the transpose of (\ref{eq:fixeddecomp}), we may assume $E \leq 0$.

We then define
\[\beta = \frac {\ln (\nnz(M) / q)}{6\ln q} \cdot 
\min \left\{1, \frac{-4E}{E + \ln q}\right\}.\]
$\beta$ can be thought of as a rescaling of $E$ into the range of $\alpha_1$ and $\alpha_2$: if $E=0$ then $\beta = 0$, while as $E$ gets more negative, $\beta$ grows. (One can verify that any decomposition yielding $\alpha_1 < \alpha_2$ and $E < 0$ has $\beta > 0$.) Hence, $\beta$ is also larger as the decomposition (\ref{eq:fixeddecomp}) is more `imbalanced'.
Finally we can formally define:
\begin{definition} \label{def:imba}
We say the decomposition (\ref{eq:fixeddecomp}) is \emph{imbalanced} if $\beta > \alpha_2 - \alpha_1$.
\end{definition}

When designing a depth-2 circuit for a matrix, one typically aims for it to be `balanced', with the same number of wires at each of the two levels of the circuit. 
Interestingly, to design such a balanced circuit for $M^{\otimes n}$, we critically need the decomposition (\ref{eq:fixeddecomp}) to be \emph{imbalanced}. This is in contrast with prior work~\cite{Alm21} which had an intentional balancing step for any input decomposition before making use of it. At a high level, imbalanced decompositions are useful for `correcting' circuits which are imbalanced in the opposite direction, whereas balanced decompositions can't be used to make such corrections. See Section~\ref{sec:overview} below for more details.

\paragraph{Application: Kronecker power matrices}

We first show that for \emph{any} $q \times q$ matrix $M \in \F^{q \times q}$, one can take a big enough Kronecker power which has a nontrivial decomposition (\ref{eq:fixeddecomp}). Our decomposition comes from a rigidity upper bound for powers $M^{\otimes n}$ and hence Theorem~\ref{thm:mainintro} applies even when $M$ is not symmetric. 

\begin{theorem}
For every integer $q \geq 2$, there is a positive constant $a_q = \Omega\left( \frac{1}{q^2 \log q} \right) > 0$ such that: For any field $\F$, any $q \times q$ matrix $M \in \F^{q \times q}$, and any positive integer $n$, the $N \times N$ matrix $M^{\otimes n}$ for $N = q^n$ has a depth-2 synchronous circuit of size $O(N^{1.5 - a_q})$.
\end{theorem}

Such a result was previously only known in the case $q=2$~\cite{Alm21}; no bound $a_q > 0$ was known even for $q=3$. As mentioned, we prove this by making use of a rigidity upper bound for $M^{\otimes n}$. The fact that these matrices are not rigid for high rank $N / \log\log N$ was recently proved by Alman~\cite{Alm21} and improved by Kivva~\cite{kivva2021improved}, based on ways to factor these matrices in terms of simpler matrices. Here we make use of a different rigidity upper bound for lower rank $N^{0.4}$, which follows more closely the proof technique of Alman and Williams~\cite{AW17} and makes use of the polynomial method. 

Notably, although the prior work~\cite{Alm21} also used rigidity upper bounds to construct circuits for Kronecker powers, the resulting circuit size degraded considerably when one used it in conjunction with a super-constant-rank rigidity upper bound. For this reason, the prior work was unable to give a nontrivial construction when $q \geq 3$, whereas our Theorem~\ref{thm:mainintro} is able to make use of these rigidity upper bounds for rank $N^{0.4}$.

\paragraph{Application: The case $q=2$} In the case $q=2$, Alman~\cite{Alm21} proved that $a_2 > 0.007$, i.e., that if $M \in \F^{2 \times 2}$, then $M^{\otimes n}$ has a depth-2 circuit of size $O(N^{1.493})$. We improve this to $a_q > 0.054$:

\begin{theorem} \label{thm:q2}
For any field $\F$, any $2 \times 2$ matrix $M \in \F^{2 \times 2}$, and any positive integer $n$, the $N \times N$ matrix $M^{\otimes n}$ for $N = 2^n$ has a depth-2 synchronous circuit of size $O(N^{1.446})$.
\end{theorem}

Both our construction and the previous construction of Alman make use of upper bounds on the rank-1 rigidity of $M^{\otimes k}$ for small positive integers $k$. Alman computes $\mathcal{R}_{M^{\otimes k}}(1)$ for $k \leq 4$ and uses this in his construction. We generalize this rigidity upper bound, and ultimately use a construction for $k=6$. That said, most of our improvement comes from the improvement of Theorem~\ref{thm:mainintro} over the approach of Alman, rather than from the new rigidity upper bound.

\paragraph{Application: The Walsh-Hadamard transform $H_n$}

The Walsh-Hadamard transform $H_n$ is a special case of Theorem~\ref{thm:q2} where we are able to give a slightly improved construction:

\begin{theorem}
\label{thm:WHLCintro}
For any field $\F$, and any positive integer $n$, the $N \times N$ matrix $H_n$ for $N = 2^n$ has a depth-2 synchronous circuit of size $O(N^{1.443})$.
\end{theorem}

The previous work by Alman~\cite{Alm21} achieved size $O(N^{1.476})$, and we show in Subsection~\ref{subsec:comparison} below that its approach could not give a better construction than that. Again, most of our improvement comes from our new Theorem~\ref{thm:mainintro}, but we also make use of the a new construction showing that $\mathcal{R}_{H_6}(1)$ is slightly lower than our upper bound on $\mathcal{R}_{M^{\otimes 6}}(1)$ for a generic matrix $M \in \F^{2 \times 2}$.

\paragraph{Application: The disjointness matrix $R_n$}

Because of the sparsity of $R_n$, its depth-2 circuit constructed from the mixed product property has size only $O(N^{1.293})$. Jukna and Sergeev~\cite{JS13} used a clever grouping of the terms of $R_1$ to give an improved construction with size $O(N^{1.272})$. We further improve this:

\begin{theorem}
For any field $\F$, and any positive integer $n$, the $N \times N$ matrix $R_n$ for $N = 2^n$ has a depth-2 synchronous circuit of size $O(N^{1.258})$.
\end{theorem}

Moreover, both the circuit of Jukna and Sergeev, and our new circuit, are what they call `SUM circuits'. That is to say they correspond to factorizations $R_n = A \times B$ where $A$ and $B$ are matrices whose entries are all $0$s and $1$s (and, in our case, $\nnz(A)$ and $\nnz(B)$ are both at most $O(N^{1.258})$). We leave open the exciting question of whether further improvements are possible by using $A$ and $B$ with other entries from $\F$.

The construction of Jukna and Sergeev can be seen as using a simple, special case of our Theorem~\ref{thm:mainintro} in which the matrix is symmetric and the decomposition is `one-sided'. In order to get our improvement, we use a decomposition of $R_3$ which is imbalanced but \emph{not} one-sided, and where the full generality of Theorem~\ref{thm:mainintro} is needed.

We note that designing a depth-2 circuit of size $O(N^{1.25 - \eps})$ for $R_n$ for any constant $\eps > 0$ would be exciting, since it's known that for any $2 \times 2$ matrix $M \in \F^{2 \times 2}$, the Kronecker power $M^{\otimes 4}$ has a depth-4 circuit whose size is at most a constant factor times the size of a depth-2 circuit for $R_n$~\cite[{Section~6}]{Alm21}. By comparison, the mixed product property gives a depth-4 circuit for $M^{\otimes n}$ of size $O(N^{1.25})$.

\section{Proof overview} \label{sec:overview}
In this section, we present a proof overview of our main circuit construction, Theorem \ref{thm:mainintro}. First, in Section~\ref{subsec:construction}, we give an algorithm for converting a sum-of-2-product decomposition (\ref{eq:fixeddecomp}) for $M$ into a linear circuit for $M^{\otimes n}$. Then, in Section~\ref{subsec:size}, we will prove that this circuit has size at most $O(\exp(\alpha_1))^n$ by analyzing a random process associated with this circuit.

\subsection{Circuit construction}\label{subsec:construction}
We first give a prototype of our algorithm, which omits some details which will be discussed in Section \ref{sec:main}.
Let $\F$ be a field, $q, n, J$ be positive integers, and $M \in \F^{q \times q}$ be a symmetric
matrix. Suppose $M$ has the decomposition $M = \sum_{j=1}^J U_j \times V_j$. Then, by
the bilinearity property and mixed product property of Kronecker product operation, we can write
\begin{align}
M^{\otimes n} &= (\sum_{j=1}^J U_j \times V_j)^{\otimes n}
= \sum_{x \in [J]^n} \bigotimes_{k=1}^n (U_{x[k]} \times V_{x[k]}) \notag \\
&= \sum_{x \in [J]^n} (\bigotimes_{k=1}^n U_{x[k]}) \times ( \bigotimes_{k=1}^n V_{x[k]})
=: \sum_{x \in [J]^n} U_x \times V_x  \label{eq:so2p}\\
&= \left(\begin{array}{c|c|c|c} U_{(1, 1, \cdots, 1)} & U_{(1, 1, \cdots, 2)}
& \cdots & U_{(J, J, \cdots, J)} \end{array}\right) \times
\left(\begin{array}{c} V_{(1, 1, \cdots, 1)} \\ \hline V_{(1, 1, \cdots, 2)} \\ \hline \vdots \\ \hline V_{(J, J, \cdots, J)} \end{array}\right) \label{eq:d2lc}
\end{align}
In Equation (\ref{eq:so2p}), we have defined $U_x := \bigotimes_{k=1}^n U_{x[k]}$ for $x \in [J]^n$, and defined $V_x$ similarly.

The calculation above shows that $M^{\otimes n}$, after being properly rewritten, can be represented in Form (\ref{eq:d2lc}), which is indeed a depth-2 linear circuit with layer 1 and 2 composed of Kronecker products of respectively $V_j$'s and $U_j$'s.
Now we try to analyze the size $Z$ of the linear circuit, i.e., 
$Z =\sum_{x \in [J]^n} \nnz(U_x) + \sum_{x \in [J]^n} \nnz(V_x)$.
By the property $\nnz(\bigotimes_{k=1}^n A_j) = \prod_{k=1}^n \nnz(A_j)$, we have
$Z = \sum_{x \in [J]^n} \prod_{k=1}^n \nnz(U_{x[k]})
+ \sum_{x \in [J]^n} \prod_{k=1}^n \nnz(V_{x[k]}) =
(\sum_{j=1}^J \nnz(U_j))^n + (\sum_{j=1}^J \nnz(V_j))^n$.
This gives us a desired bound on the size $Z$ if the two terms are `balanced', i.e., the sizes of two layers in the linear circuit, $\sum_{j=1}^J \nnz(U_j)$ and $\sum_{j=1}^J \nnz(V_j)$, are equal to each other.
For instance, if it holds that \emph{$\nnz(U_x) = \nnz(V_x)$ for all $x$}, then
we have the size bound
$Z = 2 \sum_{x \in [J]^n} \prod_{k=1}^n \sqrt {\nnz(U_{x[k]}) \cdot \nnz(V_{x[k]})}
= 2 (\sum_{j=1}^J \sqrt{\nnz(U_j) \cdot \nnz(V_j)})^n$,
which is precisely our goal in Theorem \ref{thm:mainintro}. On the other hand, when $\nnz(U_x)$ and $\nnz(V_x)$ are frequently very far apart, the size $Z$ can be much larger than this.

Thereby our primary objective here becomes: keeping each term in
the sum-of-2-product representation (Form (\ref{eq:so2p})) of $M^{\otimes n}$ (almost) balanced. This is impossible using only a single decomposition of $M$, but we will accomplish it by modifying this construction, exploiting the power of other decompositions of $M$.

We will construct our sum-of-2-product representation of $M^{\otimes n}$ iteratively by, for each $k$ from $1$ to $n$, giving a representation for $M^{\otimes k}$ in terms of our representation of $M^{\otimes k-1}$. We will particularly make use of `update functions' which, for two matrices $A$ and $B$, and any decomposition $M = \sum_{j=1}^J C_j \times D_j$ of $M$, compute $\CE(A, B) = \{(A \otimes C_j, B \otimes D_j) | j\in [J]\}$.
More formally, if $M^{\otimes (k-1)}$ has sum-of-2-product representation
$M^{\otimes (k-1)} = \sum_{s = 1}^{S} A_s \times B_s$ for some integer $S$, we will pick $S$ update functions $\CE_1, \ldots, \CE_S$ corresponding to (possibly different) decompositions of $M$, and then convert from the set $F_{k-1}$ of parts of the representation of $M^{\otimes (k-1)}$ to the set $F_k$ of parts of the representation of $M^{\otimes k}$ via
\[
\begin{array}{rccccccccc}
F_{k-1} =& \{(A_1, B_1)\} &\cup &\{(A_2, B_2)\} &\cup & \cdots & \cup
& \{(A_{S-1}, B_{S-1})\} & \cup & \{(A_S, B_S)\} \\
&\downarrow     &     &\downarrow    &     & \downarrow&   & \downarrow
&  &\downarrow\\
F_k =& \CE_1(A_1, B_1) &\cup &\CE_2(A_2, B_2) &\cup & \cdots & \cup
& \CE_{S-1}(A_{S-1}, B_{S-1}) & \cup & \CE_S(A_S, B_S).
\end{array}
\]
We can verify by the distributive property that if
$M^{\otimes (k-1)} = \sum_{(A, B)\in F_{k-1}} A \times B$, then
$M^{\otimes k} = \sum_{(A, B)\in F_k} $ $A \times B$.

We will decide which update function to use for each $(A_s, B_s)$ term in order to maintain the balance. When $\nnz(A_s) > \nnz(B_s)$, we will pick $\CE_s$ corresponding to a decomposition of $M$ which assigns $A_s$ with smaller $C_j$ and $B_s$ with larger $D_j$, and vice versa. This is possible since $M$ is symmetric, and so we always have the two update functions $\CU(A, B) = \{(A \otimes U_j, B \otimes V_j) | j\in [J]\}$ and $\CU'(A, B) = \{(A \otimes V_j^T, B \otimes U_j^T) | j\in [J]\}$ corresponding to the decompositions $M = \sum_{j=1}^J U_j \times V_j$ and $M = \sum_{j=1}^J V_j^T \times U_j^T$, respectively. Since our original decomposition is imbalanced, these two updates are imbalanced in opposite directions. Hence, we are able to make
\emph{heterogeneous} updates $\CU$ or $\CU'$ on different components in the sum-of-2-product representation of $M^{\otimes (k-1)}$, which produces a more balanced sum-of-2-product representation of $M^{\otimes k}$ than the one relying solely on $\CU$ updates.

However, only utilizing the two decompositions above is still insufficient for balancing all terms
in the sum-of-2-product representation of $M^{\otimes n}$.
For example, there could be one component $V_j^T \times U_j^T$ that is heavier on layer 1, i.e., $\nnz(V_j^T) > \nnz(U_j^T)$, even though $M = \sum_{j=1}^J V_j^T \times U_j^T$ is overall heavier on layer 2.
As a result, the sum-of-2-product representation of $M^{\otimes n}$ inevitably will contain some terms
that are not balanced using the updating procedure we just described.

To get around this, we will take advantage of two more
update functions $\CV(A, B) = \{(A \otimes I, B \otimes M)\}$ and 
$\CV'(A, B) = \{(A \otimes M, B \otimes I)\}$ corresponding to the trivial decompositions $M = I \times M$ and $M = M \times I$. These have the property that every component inside the decomposition is heavier on the same layer as is the decomposition itself. On the other hand, we want to avoid using these update functions too often, as they lead to larger circuits than our original decomposition. 
These updates therefore work as means of `\emph{hard-balancing}', as opposed to `\emph{soft-balancing}' updates $\CU, \CU'$.

We have now accumulated all the ingredients for establishing our main algorithm,
which generates a sum-of-2-product representation, or equivalently a depth-2 linear circuit of $M^{\otimes n}$.
In this algorithm, we first apply soft-balancing updates. If picked appropriately
(from $\CU, \CU'$)
according to the current size relationship between two matrices in a term, these updates are sufficient for keeping most terms balanced. However, throughout the process, some terms may become `overly unbalanced', after which we have no choice but to apply hard-balancing updates to remedy the size disparity inside; we apply hard-balancing to these terms at all remaining steps of the algorithm. Recall that our ultimate goal is to keep each term in the sum-of-2-product representation of $M^{\otimes n}$ balanced.
Hence the criterion for `overly balanced' is picked such that
a term can be precisely rebalanced by the subsequent hard-balancing updates.
More specifically, we say a term
$A \times B$ at step $k$ is overly unbalanced (say, heavier on layer 1)
if $\frac{\nnz(A)}{\nnz(B)} \cdot \left(\frac{\nnz(M)}{\nnz(I)}\right)^{n-k} \approx 1$.

Following the same notation as above, our algorithm can alternatively be seen as
\[
\begin{array}{rccccccccc}
F_{k-1} =& \{(A_1, B_1)\} &\cup &\{(A_2, B_2)\} &\cup & \cdots & \cup
& \{(A_{S-1}, A_{S-1})\} & \cup & \{(A_S, B_S)\} \\
&\downarrow     &     &\downarrow    &     & \downarrow&   & \downarrow
&  &\downarrow\\
F_k =& \CU(A_1, B_1) &\cup &\CU'(A_2, B_2) &\cup & \cdots & \cup
& \CV(A_{S-1}, B_{S-1}) & \cup & \CV'(A_S, B_S)
\end{array}
\]
where
\begin{equation}
\label{equ:conditions}
(A_s, B_s) \longrightarrow \left\{
\begin{array}{ll}
\CU(A_s, B_s),    &  \text{ if } \nnz(A_s) \ge \nnz(B_s)
\text{ and } (A_s, B_s) \text{ is in soft-balancing stage},\\
\CU'(A_s, B_s),    &  \text{ if } \nnz(A_s) < \nnz(B_s)
\text{ and } (A_s, B_s) \text{ is in soft-balancing stage},\\
\CV(A_s, B_s),    &  \text{ if } \nnz(A_s) \ge \nnz(B_s)
\text{ and } (A_s, B_s) \text{ is in hard-balancing stage},\\
\CV'(A_s, B_s),    &  \text{ if } \nnz(A_s) < \nnz(B_s)
\text{ and } (A_s, B_s) \text{ is in hard-balancing stage}.
\end{array}
\right.\end{equation}
\subsection{Circuit size analysis}\label{subsec:size}
In order to analyze the algorithm above, we capture the sizes of terms in the sum-of-2-product representation of $M^{\otimes k}$ and their changes, using a random process.
As is shown in the algorithm, to generate pairs of matrices $(A', B') \in F_k$
from $(A, B) \in F_{k-1}$, one first decides the update rule to use according to the conditions given in (\ref{equ:conditions}), then applies that update to $(A, B)$ to get a \emph{set} of pairs of matrices. In our random process, we will pick one of the pairs from that set uniformly at random, so our random process will move from one pair $(A, B) \in F_{k-1}$ to one pair $(A', B') \in F_k$.

For example, if
$(A, B)$ is to be updated by $\CU$, one needs to pick a $j \in [J]$ uniformly at random, corresponding to $U_j \times V_j$,
and construct $(A', B') \in F_k$ as $(A', B') = (A \otimes U_j, B\otimes V_j)$.
Therefore, each term in $F_k$ in fact corresponds to a sequence of $k$ iterative choices,
which we formulate as the outcomes of a sequential sampling of random variables.

Let $\TX, \TX', \TY, \TY'$ be probability distributions on $\R^2$ defined by
\begin{align*}
&\Pr[\TX = (\nnz(U_j), \nnz(V_j))] = \Pr[\TX' = (\nnz(V_j^T), \nnz(U_j^T))] = 1/J, j \in [J]; \\
&\Pr[\TY = (\nnz(I), \nnz(M))] = \Pr[\TY' = (\nnz(M), \nnz(I))] = 1.
\end{align*}
The change of size $(\frac{\nnz(A')}{\nnz(A)}, \frac{\nnz(B')}{\nnz(B)})$ at each step of the random process
is then given by the outcome of a single draw from the distribution corresponding to the update rule that is used.
Consequently the pair of sizes $(\nnz(A), \nnz(B))$, which corresponds to the size of a term $(A,B)$ we arrive at in the random process, is precisely captured by the product of a sequence of such outcomes. Further, the sum of such products over all possible outcome sequences equals the size of the entire circuit.
Therefore, to determine the size of the circuit, we define $X, X', Y, Y'$ to be the distributions $\TX, \TX', \TY, \TY'$, but with outcomes appropriately normalized so that our goal is to compute the expected value of the outcome of the following process. (We use $\odot$ to denote entry-wise product of vectors.)
\begin{center}
\fbox{\parbox{0.8\textwidth}{
\textbf{Stage 0}: \\
Draw i.i.d. examples $X_k \sim X, X_k' \sim X', Y_k \sim Y, Y_k' \sim Y', k = 1, \cdots, n$. \\ Let $S_0 = (1,1)$.

\noindent \textbf{Stage 1 (Soft-Balancing)}:

For $k = 1, 2, \cdots, n$:

~~1. Let $T_{k-1} = S_{k-1}[1] / S_{k-1}[2]$.

~~2. If $T_{k-1} \ge 1$, $S_k = S_{k-1} \odot X_k$. Otherwise, $S_k = S_{k-1} \odot X_k'$.

~~3. If $\max\{T_k, 1/T_k\} \ge \Gamma_k$, end Stage 1.

\noindent \textbf{Stage 2 (Hard-Balancing)}:

Let $k = K$ be the last iteration in Stage 1. 

For $k = K+1, K+2, \cdots, n$:

~~If $T_{k-1} \ge 1$, $S_k = S_{k-1} \odot Y_k$. Otherwise, $S_k = S_{k-1} \odot Y_k'$.
}}
\end{center}
For simplicity in this overview, we do not yet formally define the random variables and vectors involved, but their correspondence to quantities associated with our algorithm is as follows.
\[\begin{array}{|l|l|}
\hline
S_k & \text{normalized pair of sizes of some term }(A, B) \\ 
X, X', Y, Y' & \text{change of normalized size via update }\CU, \CU', \CV, \CV' \\
T_{k} \ge 1 & (A, B) \text{ is heavier on layer 2} \\
\max \{T_k, 1/T_k\} \ge \Gamma_k & (A, B) \text{ is overly unbalanced} \\
\hline
\end{array}\]
One can see from the correspondence that the random process above follows precisely the rules of our algorithm, and as a result, the expectation of $S_n$ (seen as an analogue to the sum of products over outcome sequences mentioned above), equals the size of the final circuit (\ref{eq:d2lc}).
In this way we manage to reduce the task of showing the desired circuit size bound to
an examination of the random process above.

Unfortunately this process involves a changing two-dimensional random vector ($S_k$) which is hard to analyze directly, due to the simultaneous changing of $\nnz(A)$ and $\nnz(B)$.
We therefore turn to another similar process whose outcome has the same expected value, and in which the key variable is only one-dimensional.

For each $j \in [J]$, let $K_j := \sqrt{\nnz(U_j) \cdot \nnz(V_j)}$. In the process above, the distribution $X$ returns $(\nnz(U_j), \nnz(V_j))$ with probability $1/J$ for each $j \in [J]$. Suppose we modify it to instead return $(\nnz(U_j) / K_j, \nnz(V_j) / K_j)$ with probability $K_j / (\sum_{j' = 1}^J K_{j'})$ for each $j \in [J]$. We can verify that this does \emph{not} change the expected value of the result of the random process above. On the other hand, this gives us the additional invariant that 
$(S_{k+1}[1] \cdot S_{k+1}[2])/(S_{k}[1] \cdot S_{k}[2])$ is always a constant. Therefore, the two-dimensional vector $S_k$ is now uniquely determined by the single real number $T_k$, which is easier for us to analyze. Interestingly, the new process now corresponds less directly to the original algorithm (one can think of it as making `fractional' choices of parts of the decomposition of $M$ at each step), but is substantially easier to analyze using the usual toolbox.

\subsection{Applications}
With Theorem \ref{thm:mainintro} proved, it remains to find a sparse (nontrivial) decomposition which is either imbalanced or one-sided in order to construct small linear circuits. As mentioned above, we are able to find such a decomposition for Kronecker powers of \emph{any} fixed matrix $M$.
\paragraph{Rigidity-based decompositions}
By the definition of rigidity, for any rank $r$, the matrix $M\in \F^{q\times q}$  has a decomposition $M = U \times V + I \times S$ 
for some $U \in \F^{q \times r}, V \in \F^{r \times q}$ (which both have at most $q \cdot r$ nonzero entries), and where $\nnz(S) = \mathcal{R}_M(r)$.
This one-sided decomposition, in conjunction with our Theorem~\ref{thm:mainintro}, gives
a depth-2 linear circuit of size $O((q \cdot r + \sqrt{q \cdot \rig_M(r)})^n)$ for $M^{\otimes n}$. (By comparison, the prior work~\cite{Alm21} achieved a depth-2 circuit of size $O(( \sqrt{(qr+q)(qr + \mathcal{R}_M(r))})^n)$ from the same rigidity bound; one can verify by the Cauchy-Schwarz inequality that our new bound strictly improves on this.)
That is, small linear circuits for $M^{\otimes n}$ follow from rigidity upper bounds of $M$. We then make use of and generalize various recent techniques to give the desired rigidity upper bounds.
\begin{itemize}
\item Alman \cite[Subsection 4.1 and 4.2]{Alm21} presented (ad hoc) rank-1 rigidity bounds on specific matrices of fixed size, including the third Kronecker power of any $2 \times 2$ matrix, as well as the Walsh-Hadamard matrix $H_4$. In this work, new rank-1 rigidity bounds shown in Theorem \ref{thm:rigKron} and \ref{thm:rigWH}
generalize Alman's results to $M^{\otimes n}$ and $H_n$ for arbitrary integer $n$ using an alternative scheme of matrix indexing.
\item The polynomial method is another approach for proving rigidity upper bounds, which has been used to show a number of matrices of interest are not rigid enough to prove lower bounds using Valiant's approach~\cite{valiant1977graph} (e.g., \cite{AW17, DL19, Alm21}).
Using our construction, we know that for $M \in \F^{q \times q}$, there exists a non-trivial depth-2 circuit for $M^{\otimes n}$ if 
$\rig_{M_n}(r) < (q^n)^2$ for some $r < (q^n)^{0.5}$ and any $n$.
Partly inspired by \cite[Theorem A.1]{AW17}, we manage to show a result of this type for any fixed matrix $M \in \F^{q \times q}$ in Theorem \ref{thm:rigGKron} with the aid of the polynomial method.
\end{itemize}

\paragraph{Partition-based decompositions}
Any disjointness matrix $R_n$, as a special Kronecker power matrix,
automatically has the decompositions stated above. 
However further exploiting the fact that $R_n$ is sparse and contains only 0 and 1 entries, one can
construct other tailor-made decompositions for $R_n$ based on partitioning its 1-entries into small blocks (`combinatorial rectangles'). The prior best bound given by Jukna and Sergeev~\cite{JS13}
implicitly relies on a 2-partition of $R_1$; we improve on this
by presenting a (slightly) improved 8-partition for $R_3$.
We also remark that converting this new decomposition into a linear circuit is beyond the capability of Jukna and Sergeev's approach, which works only for one-sided decompositions, and necessitates our main result Theorem \ref{thm:mainintro}.

\section{Preliminaries}
\subsection{Notations and conventions}
\label{subsec:notation}

We begin by introducing some notation and conventions that we will use throughout this paper.
\paragraph{Notation}
$\exp : \R^{n} \rightarrow \R^{n}$ is the entry-wise exponentiation function defined as $\exp(x[1], x[2], \cdots, $ $x[n]) = 
(e^{x[1]}, e^{x[2]}, \cdots, e^{x[n]})$. When $n = 1$, it coincides with the normal exponentiation of reals.
$\odot$ is the binary entry-wise product (also known as Hadamard product)
between vectors of same dimension, i.e.,
$(x[1], \cdots, x[n]) \odot (y[1], \cdots, y[n]) = (x_1[n] \cdot y_1[n], 
\cdots x[n] \cdot y[n])$.
$\Ex[X]$ for random vector $X = (X[1], \cdots, X[n])$ on $\R^n$ is the entry-wise expectation, i.e. $\Ex[X] = (\Ex[X[1]], \cdots, \Ex[X[n]])$.
$\log_a(b)$ is the logarithm of $b > 0$ with base $a > 0, a\ne 1$.
In particular, we use $\log(\cdot)$ to denote $\log_2(\cdot)$ and $\ln(\cdot)$ to denote $\log_{e}(\cdot)$.
$\bo[P]$ for some property $P$ is the indicator function
\[\bo[P] = \left\{\begin{array}{cc} 1, & \text{if $P$ is true}, \\ 
0, & \text{if $P$ is false}. \end{array} \right.\]
$\ip{x}{y}$ is the inner product between vectors $x, y \in \R^{n}$ over the reals. $\varnothing$ denotes the empty set.

\paragraph{Conventions}
For integer $z$, $0^z = \bo[z = 0]$.
$\min(\varnothing) = +\infty$, $\max(\varnothing) = -\infty$.

\subsection{Kronecker product and matrix sparsity}
For two matrices $A \in \F^{a_1 \times a_2}$ and $B \in \F^{a_3 \times a_4}$, their \emph{Kronecker product} is the matrix $A \otimes B \in \F^{a_1a_3 \times a_2a_4}$ given by, for $i_1 \in [a_1]$, $i_2 \in [a_2]$, 
$i_3 \in [a_3]$ and $i_4 \in [a_4]$,
$$A \otimes B[i_1 + a_1 i_3, i_2 + a_2 i_4]  = A[i_1, i_2] \cdot B[i_3, i_4].$$
We in addition show some properties of Kronecker product operation of particular use here.
\begin{enumerate}
\item (bilinearity) for any matrices $A, B \in \F^{a_1 \times a_2}, C, D \in \F^{b_1 \times b_2}$,
$(A + B) \otimes (C + D) = A \otimes C + A \otimes D + B \otimes C + B \otimes D$,
\item (mixed-product property) for any matrices $A \in \F^{a_1 \times a_2}, B \in \F^{b_1 \times b_2}, C \in \F^{c_1 \times c_2}, D \in \F^{d_1 \times d_2}$ with $a_2 = c_1$ and $b_2 = d_1$, 
$(A \otimes B) \times (C \otimes D) = (A \times C) \otimes (B \times D)$.
\item (rank) for any matrices $A \in \F^{a_1 \times a_2}, B \in \F^{b_1 \times b_2}$, 
$\rank(A \otimes B) = \rank(A) \cdot \rank(B)$.
\end{enumerate}

For a matrix $A \in \F^{a_1 \times a_2}$ on field $\F$, we use $\nnz(A)$ to denote its \emph{sparsity}, i.e. the number of non-zero entries in $A$.
Some basic properties we will use are that, for any $A \in \F^{a_1 \times a_2}$
and $B \in \F^{b_1 \times b_2}$,
\begin{enumerate}
\item $\nnz(A \otimes B) = \nnz(A) \cdot \nnz(B)$,
\item if $D \in \F^{a_1 \times a_2}$ is a diagonal matrix, then $\nnz(D \times A) \le \nnz(A)$.
\end{enumerate}

\subsection{Matrix rigidity}
\label{subsec:prelim-rig}
Let $\F$ be a field, and $\mathbb K$ be an extension field of $\F$.
For a matrix $A \in \F^{a_1 \times a_2}$ and a non-negative integer $r$, we write
$\rig_A^{\mathbb K}(r)$ to denote the \emph{rank-$r$ rigidity} of $A$ over $\mathbb K$, which is the minimum number of entries of $A$ which must be changed to other values in $\mathbb K$ to make its rank at most $r$. In other words,
\[\rig_A^{\mathbb K}(r) := \min_{\begin{substack}{ B \in \mathbb K^{a_1 \times a_2} \\ \rank(A + B) \le r}
\end{substack}} \nnz(B).\]
We note that the definition of $\rig_A^{\mathbb K}(r)$ depends on the field $\mathbb K$.
In particular, Babai and Kivva~\cite{BK21} find fields $\F, \K$ such that $\mathbb K$ is an extension field of $\F$,
and a matrix $M \in \F^{n\times n}$ such that
for some rank $0 < r < n$, $\rig_{M_n}^{\mathbb K}(r)$ is strictly less than $\rig_{M_n}^{\F}(r)$.
Nonetheless, in this work we will only focus on the rigidity of $M \in \F^{n\times n}$ over the field $\F$ where we are designing a linear circuit, and our bounds work equally well over any field. In the following discussion we will use $\rig_{M}(r) := \rig_M^{\F}(r)$ to denote the rigidity if $M \in \F^{n \times n}$, and $\mathbb K= \F$.

By the definition of rank-$r$ rigidity, $A$ can always be rewritten as
$A = L + S$ for some matrices $L, S \in \F^{a_1 \times a_2}$, 
where $\rank(L) \le r, \nnz(S) \le \rig_A(r)$.
Further, there exist matrices $U \in \F^{a_1 \times r}, V \in \F^{r \times a_2},
S \in \F^{a_1 \times a_2}$ such that $L = U \times V$ and hence $A = U \times V + S$.

Initially matrix rigidity was developed as a tool for proving low-depth circuit lower bounds.
Valiant~\cite{valiant1977graph} showed that if the $N \times N$ matrix $M$ has $\mathcal{R}_M(N / \log \log N) \geq N^{1 + \eps}$ for any constant $\eps>0$, then $M$ doesn't have a circuit of depth $O(\log N)$ and size $O(N)$.
In light of this result, a matrix $M$ satisfying this rigidity lower bound is frequently called \emph{Valiant-rigid} in the literature.
Towards finding an explicit construction of Valiant-rigid matrices,
many researchers have subsequently shown various rigidity lower bounds, e.g. 
\cite{SSS1997riglb, DBLP:Fri93, midrijanis, Lok00}. The best known rigidity lower bound for the $N \times N$ Walsh-Hadamard transform $H_n$ (for $N = 2^n$) is $$\rig_{H_n}(r) \geq \frac{N^2}{4r}.$$

This is not sufficient to prove that $H_n$ is Valiant-rigid, and despite much effort, there are no known `fully explicit' constructions of Valiant-rigid matrices.
A recent line of work instead showed that a number of families of explicit matrices are in fact not Valiant-rigid~\cite{AW17,dvir2019matrix,DL19,Alm21, kivva2021improved,bhargava2022fast}, including
the Walsh-Hadamard transform \cite{AW17}, discrete Fourier transform \cite{DL19}, Kronecker product matrices \cite{Alm21, kivva2021improved} and certain Vandermonde matrices~\cite{bhargava2022fast}. It is not immediate that these rigidity upper bounds lead to improved linear circuits. Nonetheless, in this paper, we use rigidity upper bounds to construct linear circuits.

As stated before, our new construction of depth-2 circuits for $H_n$ and $M_n$ are both based on rigidity upper bounds. We note that Theorem \ref{thm:rigWH} and \ref{thm:rigGKron} essentially provide new rigidity upper bounds for such matrices respectively in regime $r = 1$ and $r = O(N^{0.4})$,  different than the regime $r = N/\log\log N$ of interest in Valiant-rigidity. These results about the structure of Kronecker power matrices may be of independent interest.

\subsection{Binomial coefficients and binary entropy function}
The binary entropy function $H : [0, 1] \rightarrow [0, 1]$ is defined by
\[H(p) = -p\log p - (1-p)\log (1-p)\]
where we take $0 \cdot \log(0) = 0$. For every integer $n > 1$ and every $p \in (0, 1)$,
it is known that
\[\frac{1}{n+1} \cdot 2^{n \cdot H(p)} \le \binom {n}{pn} 
\le 2^{n \cdot H(p)}.\]

By Taylor's expansion, we further have the following bounds.
\begin{lemma}
\label{lem:entropy}
For $0 < p < 1/2$,
\begin{enumerate}
\item $H(p + \delta) \le H(p) + H'(p) \cdot \delta + H''(p + \delta) \cdot \frac{\delta^2}2$,
\item $H(p - \delta) \le H(p) - H'(p) \cdot \delta + H''(p) \cdot \frac{\delta^2}2$.
\end{enumerate}
Here $H'(p) = \log_2(\frac{1- p}{p})$, $H''(p) = -\frac {\log_2 e}{p(1-p)}$
are respectively the first- and second-order derivative of $H(p)$.
\end{lemma}

In addition, we will make use of the following identity regarding binomial coefficients.
For positive integers $n, s$, and nonnegative integer $r$ with $r, s\le n$, a multisection of the binomial expansion is defined as
\[\MS(n, s, r) := \sum_{
\begin{substack}{ t \equiv r ~ (\mathrm{mod}~s) \\ 0 \le t \le n}
\end{substack}} \binom n t. \]

\begin{lemma}[Series multisection identity, \cite{multisect}]
\label{lem:multisection}
\[\MS(n, s, r) = \frac 1s \sum_{j=0}^{s-1} \left( 2\cos \frac{\pi j}{s} \right)^n
\cos \frac{\pi (n-2r) j}{s}.\]
In particular,
\[\MS(n, 4, \frac n2 - r') = \frac12\left(2^{n-1} + 2^{n/2} \cos\frac{r'\pi}{2}\right).\]
\end{lemma}

\subsection{Chernoff bound}
\begin{theorem}
\label{thm:Chernoff}
$X_i, i \in [n]$ are i.i.d. random variables distributed in $[0, 1]$ with expectation $\Ex[X_i] = \mu$. Let $X = \sum_{i=1}^n X_i$. Then for $\delta > 0$,
\[\Pr[X > (1 + \delta)\mu n] \le \left\{\begin{array}{cc}
\exp(-\delta^2 \mu n / 3), & \text{ if } \delta \le 1, \\
\exp(-\delta \mu n / 3), & \text{ if } \delta > 1. \\
\end{array}\right.\]
\end{theorem}
\subsection{Tools for rigidity}
\begin{definition}
Suppose $\F$ is a field, and $q$ is a positive integer.
A matrix $M \in \F^{q \times q}$ is called an outer-1 matrix if for any $i, j \in \{0, 1, \cdots, q\}$,
$M[i, j] = 1$ whenever either $i = 0$ or $j = 0$. 
We similarly say $M$ is an outer-nonzero matrix if we have
$M[i, j] \ne 0$ for such $i, j$.
\end{definition}
\begin{lemma}[\cite{Alm21}, Lemma 2.8]
\label{lem:outer1}
For any field $\F$, positive integer $q$, and outer-nonzero matrix $M \in \F^{q \times q}$,
there are 
an outer-1 matrix $M' \in \F^{q \times q}$, and
two invertible diagonal matrices $D, D' \in \F^{q \times q}$,
such that $M = D \times M' \times D'$.
\end{lemma}
\begin{lemma}[\cite{Alm21}, Lemma 2.10]
For any field $\F$, positive integers $q, r$, and matrices $A, B, D, D' \in \F^{q \times q}$ such that
$D$ and $D'$ are invertible diagonal matrices with $A = D \times B \times D'$,
we have that $\rig_A(r) = \rig_B(r)$.
\label{lem:diag}
\end{lemma}
\begin{theorem}[interpolation polynomial, \cite{AW15}]
\label{thm:interpolation}
For any field $\F$, positive integers $n, r, k$ such that $n \ge r+k$, 
and $c_1, \cdots, c_r \in \F$, 
there exists a multivariate polynomial $p : \{0, 1\}^n \rightarrow \F$ of degree $(r-1)$ with coefficients in $\F$ such that for all $i \in [r]$ and $x \in \{0, 1\}^n$ with Hamming weight $|x| = k+i$, $p(x) = c_i$.
\end{theorem}
\begin{remark}
The original statement in \cite{AW15} focuses only on the existence of interpolation polynomial with integer coefficients, assuming $c_1, \cdots, c_r \in \Z$. But as their techniques naturally generalize to any additive abelian group, and thus any field $\F$, we opt to omit a (nearly) verbatim proof.
\end{remark}
\begin{lemma}[polynomial method]
\label{lem：polymethod}
For any field $\F$, positive integers $n, s, r$, and matrix $M \in \F^{2^n \times 2^n}$,
if there exists polynomial $f: \{0, 1\}^n \times \{0, 1\}^n \rightarrow \F$ such that,
\begin{itemize}
\item $f(x, y) = M[x, y]$ for $s$ many pairs $(x, y) \in \{0, 1\}^n \times \{0, 1\}^n$, and 
\item the expansion of $f$ contains $r$ monomials,
\end{itemize}
then $\rig_{M}(r) \le 2^{2n} - s$.
\end{lemma}
\begin{proof}
A polynomial $f$ containing $r$ monomials can be written as 
$f(x, y) = \sum_{i=1}^r g_i(x) \cdot h_i(y)$, where $g_i(x), h_i(y)$ are respectively the $x$-part and $y$-part of the $i$-th monomial.
Then consider matrix $M_f \in \F^{2^n \times 2^n}$ defined by $M_f[x, y] = f(x, y)$.
If we construct matrices $U \in \F^{2^n \times r}, V \in \F^{r \times 2^n}$
with $U[x, i] = g_i(x), V[i, y] = h_i(y)$ for $x, y \in \{0, 1\}^n, i \in [r]$,
then $M_f$ can be alternatively computed by $M_f = U \times V$.
Therefore, writing $M$ as $M = M_f + (M - M_f)$, we simultaneously have
$\rank(M_f) = \rank(U \times V) \le r$ and 
$\nnz(M - M_f) \le 2^{2n} - \#[M[x, y] = f(x, y)] = 2^{2n} - s$,
which by the definition of rigidity implies $\rig_M(r) \le 2^{2n} - s$.
\end{proof}

\section{Main Theorem}
\label{sec:main}
In this section, we will give the full proof of Theorem \ref{thm:mainintro}, which is restated as follows.
\paragraph{Setup}
Suppose $\F$ is a field, $q, J$ are positive integers, and
$M \in \F^{q \times q}$ is a symmetric matrix which has the decomposition
\begin{align} \label{eq:gendecomp} M = \sum_{j=1}^J U_j \times V_j. \end{align}
where for every $j \in [J]$, matrices $U_j \in \F^{q \times r_i}, V_j \in \F^{r_i \times q}$
for some positive integer $r_i$.\footnote{
By partitioning the matrices by rows and columns, 
we can always convert decomposition (\ref{eq:gendecomp}) to be of the same form as decomposition (\ref{eq:fixeddecomp}),
i.e. with $r_i = 1, ~\forall i \in [J]$.
It can be verified that they generate depth-2 circuits of the same size
via Theorem \ref{thm:mainintro}.}

Some parameters associated with this decomposition
are defined as follows.
\footnote{$I$, if not specified otherwise, is the identity matrix of size $q \times q$.}
\[G := \max \left\{\ln\left(\nnz(M) / \nnz(I)\right), \max_{j \in [J]} \left|\ln\left(\nnz(U_j)/\nnz(V_j)\right)\right|\right\},\]
\[E := \frac{\sum_{j=1}^J  \ln (\nnz(U_j)/ \nnz(V_j)) \cdot \sqrt{\nnz(U_j) \cdot \nnz(V_j)}}{\sum_{j=1}^J \sqrt{\nnz(U_j) \cdot \nnz(V_j)}}.\]

\begin{definition}
Let
\[\alpha_1 = \ln \left(\sum_{j=1}^J \sqrt{\nnz(U_j) \cdot \nnz(V_j)}\right), \quad
\alpha_2 = \ln \left(\sqrt{\nnz(M) \cdot \nnz(I)}\right), \]
\[\beta = \frac {\ln(\nnz(M) / \nnz(I))}{6G} \cdot \min 
\left\{1, \frac{-4E}{E + G}\right\}.\]
We say decomposition $M = \sum_{j=1}^J U_j \times V_j$ is imbalanced if $\beta > \alpha_2 - \alpha_1$.
\end{definition}
\begin{theorem}
\label{thm:main}
For any positive integer $n$, if $M$ has the imbalanced decomposition above, then $M^{\otimes n}$ has a depth-2 linear circuit of size
$\exp(\alpha_1)^{n + o(n)}$.
\end{theorem}

In the following discussion, we will first present an algorithm for producing a linear circuit for $M^{\otimes n}$, then show that it has the desired size.
For matrices $A$ and $B$, define $\gamma(A, B)= \ln(\nnz(A) / \nnz(B)) \in \R$,
and let $\CU(A, B)$ and $\CV(A, B)$ be the following multisets of pairs of matrices.
\begin{align*}
\CU(A, B) &= \left\{\begin{array}{ll} 
\{(A \otimes U_j, B \otimes V_j) ~|~ j\in [J]\}, & \text{ if } \gamma(A, B) \ge 0, \\
\{(A \otimes V_j^T, B \otimes U_j^T) ~|~ j \in [J]\}, & \text{ if } \gamma(A, B) < 0.
\end{array}\right.\\
\CV(A, B) &= \left\{\begin{array}{ll} 
\{(A \otimes I, B \otimes M) \}, & \text{ if } \gamma(A, B) \ge 0, \\
\{(A \otimes M, B \otimes I) \}, & \text{ if } \gamma(A, B) < 0.
\end{array}\right.
\end{align*}
With the definition of $\CU(A, B), \CV(A, B)$, we inductively define series of collections of matrix pairs $\CF_k, \CG_k,$ $\CH_k, k= 0, 1, 2, \cdots, n$ as follows.
\begin{itemize}
\item $\CF_0 = \{(I_1, I_1)\}$, $\CG_0 = \varnothing$.
\item At step $k \in [n]$, let $\Gamma_k = (n-k) \cdot \gamma(M, I) + 2G$.
\begin{align*}
\CH_k &:= \bigcup_{(A, B) \in \CF_{k-1}} \CU(A, B), \\
\CF_k &:= \{ (A', B') ~|~ (A', B') \in \CH_k, |\gamma(A', B')| < \Gamma_k\}, \\
\CG_k &:= \left[\bigcup_{(A, B) \in \CG_{k-1}} \CV(A, B)\right] \cup (\CH_k \backslash \CF_k)
\end{align*}
\end{itemize}

Then we claim that the matrix pairs in $\CF_n \cup \CG_n$, if multiplied pairwise, correspond to a decomposition of $M^{\otimes n}$.
More formally, we have the following proposition.
\begin{proposition}
Let $\CH = \CF_n \cup \CG_n$, then $\sum_{(A,B) \in \CH} (A \times B) = M^{\otimes n}$.
\end{proposition}
\begin{proof}
As $M$ is symmetric, we simultaneously have 
$M = \sum_{j\in [J]} U_j \times V_j$,
and $M = \sum_{j \in [J]} V_j^T \times U_j^T$.  
Hence
\begin{alignat*}{3}
\sum_{(A', B') \in \CU(A, B)} (A' \times B')
&= \left\{\begin{array}{ll} 
(A \times B) \otimes \sum_{j\in[J]}(U_j \times V_j), & \text{ if } \gamma(A, B) \ge 0, \\
(A \times B) \otimes \sum_{j\in[J]}(V_j^T \times U_j^T), & \text{ if } \gamma(A, B) < 0
\end{array}\right.
&= (A \times B) \otimes M, \\
\sum_{(A', B') \in \CV(A, B)} (A' \times B')
&= \left\{\begin{array}{ll} 
(A \times B) \otimes (I \times M), & \text{ if } \gamma(A, B) \ge 0, \\
(A \times B) \otimes (M \times I), & \text{ if } \gamma(A, B) < 0
\end{array}\right.
&= (A \times B) \otimes M.
\end{alignat*}

Then for all $k \in [n]$,
\begin{align*}
\sum_{(A', B') \in \CH_k} (A' \times B')
&= \sum_{(A, B) \in \CF_{k-1}} \sum_{(A', B') \in \CU(A, B)} (A' \times B')
= \left[\sum_{(A, B) \in \CF_{k-1}}(A \times B)\right] \otimes M. \\
\sum_{(A', B') \in \CG_k} (A' \times B') &= 
\sum_{(A, B) \in \CG_{k-1}} \sum_{(A', B') \in \CV(A, B)} (A' \times B')
+ \sum_{(A', B') \in \CH_k \backslash \CF_k} (A' \times B') \\
&= \left[\sum_{(A, B) \in \CG_{k-1}}(A \times B)\right] \otimes M
+ \sum_{(A', B') \in \CH_k} (A' \times B') - \sum_{(A', B') \in \CF_k} (A' \times B').
\end{align*}
\[\Rightarrow \sum_{(A', B') \in \CG_k \cup \CF_k} (A' \times B')
= \left[\sum_{(A, B) \in \CG_{k-1} \cup \CF_{k-1}}(A \times B)\right] \otimes M.\]

By induction, we immediately have 
\[ \sum_{(A, B) \in \CG_n \cup \CF_n} (A \times B)
= \left[\sum_{(A, B) \in \CG_0 \cup \CF_0}(A \times B)\right] \otimes M^{\otimes n}
= M^{\otimes n}.\]
\end{proof}
Based on the claim above, one can construct the following depth-2 linear circuit for $M^{\otimes n}$.
\[M^{\otimes n} = \left(\begin{array}{c|c|c|c} A_1 & A_2 & \cdots & A_{|\CH|} \end{array}\right) \times
\left(\begin{array}{c} B_1 \\ \hline B_2 \\ \hline \vdots \\ \hline B_{|\CH|}\end{array}\right),\]
where $(A_i, B_i), i = 1, 2, \cdots, |\CH|$ are the properly indexed elements in $\CH$.
This circuit will be referred to in future as $\CC(M^{\otimes n}) = \CC(M^{\otimes n})_{A} \times \CC(M^{\otimes n})_{B}$ for convenience.

Let $\size(\CC(M^{\otimes n})) = \nnz(\CC(M^{\otimes n})_{A}) + \nnz(\CC(M^{\otimes n})_{B})$ denote the size of the circuit above. In the following subsections we shall give bounds on this quantity dependent on the initial decomposition, and thus complete Theorem \ref{thm:main}.

\subsection{An associated random walk}
To facilitate further discussion, we first show a connection between $\size(\CC(M^{\otimes n}))$ and a related two-stage randomized algorithm.
We define a sequence $S_k \in \R^2, k=0, 1, \cdots, n$ of two-dimensional random vectors. Initially $S_0 = (0, 0)$ with probability 1. 
Define function $t : \R^2 \rightarrow \R$ by $t(s) = s[1] - s[2]$,
and define the random variables $T_k = t(S_k), k = 0, 1, \cdots, n$.

Let $X, X', Y, Y'$ be probability distributions on $\R^2$ with
\[\Pr[X = (\ln J + \ln \nnz(U_j), \ln J + \ln \nnz(V_j))] = 1/J, j \in [J],\]
\[\Pr[X' = (\ln J + \ln \nnz(V_j^T), \ln J + \ln \nnz(U_j^T))] = 1/J, j \in [J],\]
\[\Pr[Y = (\ln \nnz(I), \ln \nnz(M))] = 1,\]
\[\Pr[Y' = (\ln \nnz(M), \ln \nnz(I))] = 1\]
(for simplicity, cases with the same value are not merged).
Then we have randomized algorithm (I) as follows.

\begin{center}
\fbox{\parbox{0.8\textwidth}{
\textbf{Stage 0}: \\
Draw i.i.d. examples $X_k \sim X, X_k' \sim X', Y_k \sim Y, Y_k' \sim Y', k = 1, \cdots, n$.

\noindent \textbf{Stage 1}:

For $k = 1, 2, \cdots, n$:

~~1. If $T_{k-1} \ge 0$, $S_k = S_{k-1} + X_k$. Otherwise, $S_k = S_{k-1} + X_k'$.

~~2. If $|T_k| \ge \Gamma_k$ (definition of $\Gamma_k$ as before), end Stage 1.

\noindent \textbf{Stage 2}:
Let $k = K$ be the last iteration in Stage 1. 

For $k = K+1, K+2, \cdots, n$:

~~If $T_{k-1} \ge 0$, $S_k = S_{k-1} + Y_k$. Otherwise, $S_k = S_{k-1} + Y_k'$.
}}
(I)
\end{center}
The connection between $\size(\CC(M^{\otimes n}))$ and algorithm (I) is captured by the following lemma. 
\begin{lemma}
\label{lem:SimpleRP}
Let $S_k, k=0, 1, \cdots, n$ be the random vectors in randomized algorithm (I). Then
\[\Ex[\exp(S_n)] 
= (\nnz(\CC(M^{\otimes n})_{A}), \nnz(\CC(M^{\otimes n})_{B})).\]
\end{lemma}
Recall that the function $\exp: \R^2 \rightarrow \R^2$ is the entry-wise exponentiation.
(See Subsection \ref{subsec:notation})

Before proving this lemma, we first take a look from an alternative perspective at the definition of $\CF_k$ and $\CG_k$. As every element in $\CF_{k} \cup \CG_{k}$ is generated by some element in $\CF_{k-1} \cup \CG_{k-1}$ via either $\CU$ or $\CV$ update, it is possible to formulate this generation as a tree $\CT$ of depth $n$ (uniformly for all branches). 
A node $(A, B)$ at depth-$k$ of $\CT$ is a matrix pair in either $\CF_k$ or $\CG_k$.
Between two nodes $(A, B), (A', B')$, there exists an edge marked as $\CU$ (resp. $\CV$)
if and only if $(A', B') \in \CU(A, B)$ (resp. $(A', B') \in \CV(A, B)$).

Further, we note that all the leaves of a subtree rooted at $(A, B)$ can be written in the form $L = (A \otimes C_L, B \otimes D_L)$ for some matrices $C_L, D_L$.
Hence there is a natural way to measure the contribution of a subtree to the overall size of $\CC(M^{\otimes n})$. Define $\sigma(A, B) = (\sum_{L \in \CT(A, B)} \nnz(C_L), $ $\sum_{L \in \CT(A, B)} \nnz(D_L))$, in which $\CT(A, B)$ is the collection of all leaves on the subtree with root $(A, B)$. When applied to the entire tree, this measure precisely captures the size of $\CC(M^{\otimes n})$. That is, $\sigma(I_1, I_1) = (\nnz(\CC(M^{\otimes n})_{A}), \nnz(\CC(M^{\otimes n})_{B}))$.

With this measure, we try to establish the connection between $\CC(M^{\otimes n})$ and $S_n$ using the following proposition. 

\begin{proposition}
\label{prop:SimpleRP}
Let $(A^{(0)}, B^{(0)}), (A^{(1)}, B^{(1)}), \cdots, (A^{(n)}, B^{(n)})$
be a path from root to leaf in \CT. 
$(A^{(i)},$ $B^{(i)})= (A^{(i-1)} \otimes C^{(i)}, B^{(i-1)} \otimes D^{(i)}), i \in [n]$. For $i \in [n]$, define
\[r_i = \left\{\begin{array}{ll}
\ln J, & \text{ if } 
(A^{(i)}, B^{(i)}) \in \CU(A^{(i-1)}, B^{(i-1)}), \\
0, & \text{ if } (A^{(i)}, B^{(i)}) \in \CV(A^{(i-1)}, B^{(i-1)}),
\end{array}\right.\quad R_i = \sum_{l=1}^{i} r_l.\]
Pick $p_i = (R_i + \ln (\nnz(A^{(i)}), R_i + \ln (\nnz(B^{(i)})), i \in [n]$. Then
\[\Ex[\exp(S_n - S_{k}) | (S_1, \cdots, S_{k}) = (p_1, \cdots, p_{k})] = 
\sigma(A^{(k)},  B^{(k)})\]
holds for any integer $0 \le k \le n$.
\end{proposition}
\begin{proof}
First we define $h((s_1, \cdots, s_{k-1})) = \min\{i \in [k-1] ~|~ |t(s_i)| \ge \Gamma_i\}$ to indicate whether step $k$ is in Stage 2. 
\footnote{We follow the convention that $\min \varnothing = \infty$.}
More specifically, consider the first $k-1$ steps with known outcome 
$(S_1, \cdots, S_{k-1}) = (s_1, \cdots, s_{k-1})$. 
If $h((s_1, \cdots, s_{k-1})) = \infty$, i.e. 
$\{i \in [k-1] ~|~ |t(s_i)| \ge \Gamma_i\} = \varnothing$, 
it means the process passes all the tests $|T_k| < \Gamma_k$ and stays in Stage 1 for step $k$.
Otherwise, for some $i \in [k-1]$ this condition is violated, causing the process to enter Stage 2 and to update following the corresponding rules for all steps $i+1, \cdots, n$, including step $k$.

Then for arbitrary $s_1, \cdots, s_{k-1}$ with $\Pr[(S_1, \cdots, S_{k-1}) = (s_1, \cdots, s_{k-1})] > 0$,
by exhausting cases with respect to $h_{k-1} := h((s_1, \cdots, s_{k-1}))$
and $t_{k-1} := t(s_{k-1})$, 
we write the probability of updates as follows.
\begin{align}
\label{eq:update_prob}
& ~~~~ \Pr[S_{k} = s_{k} | (S_1, \cdots, S_{k-1}) = (s_1, \cdots, s_{k-1})] \notag \\
&= \left\{\begin{array}{ll}
\Pr[X = s_{k} - s_{k-1}], & \text{ if } h_{k-1} = \infty, t_{k-1} \ge 0, \\
\Pr[X' = s_{k} - s_{k-1}], & \text{ if } h_{k-1} = \infty, t_{k-1} < 0, \\
\Pr[Y = s_{k} - s_{k-1}], & \text{ if }  h_{k-1} < \infty, t_{k-1} \ge 0, \\
\Pr[Y' = s_{k} - s_{k-1}], & \text{ if } h_{k-1} < \infty, t_{k-1} < 0.
\end{array}\right. 
\end{align}

Now consider the situation where we fix $s_i = p_i, i \in [k-1]$.
$t_i := t(p_i) = \gamma(A^{(i)}, B^{(i)})$ is easily obtained, while for 
$h_{k-1} := h((p_1, \cdots, p_{k-1}))$ we have the following claim.
\begin{claim}
\[h_{k-1} = \infty \Longleftrightarrow
(A^{(k)}, B^{(k)}) \in \CU(A^{(k-1)}, B^{(k-1)}),\]
\[h_{k-1} < \infty \Longleftrightarrow
(A^{(k)}, B^{(k)}) \in \CV(A^{(k-1)}, B^{(k-1)}).\]
\end{claim}
\begin{proof}
For the case $h_{k-1} = \infty$, we know $\{i \in [k-1] ~|~ |\gamma(A^{(i)}, B^{(i)})| \ge \Gamma_i\} = \varnothing$. That is to say if $(A^{(i)}, B^{(i)}) \in \CF_i$, then 
$(A^{(i+1)}, B^{(i+1)})$ falls into $\CF_{i+1}$ after one $\CU$ update. Hence $(A^{(k-1)}, B^{(k-1)}) \in \CF_{(k-1)}$ as $(A^{(0)}, B^{(0)}) \in \CF_0$. 
Further, $(A^{(k)}, B^{(k)}) \in \CU(A^{(k-1)}, B^{(k-1)})$.

For the other case, $|\gamma(A^{(i)}, B^{(i)})| < \Gamma_i$ holds for $i = 0, \cdots, h_{k-1}-1$.
(this is well-defined since $|\gamma(A^{(0)}, B^{(0)})| = 0 < \Gamma_0$, 
and thus $h_{k-1} \ge 1$).
From the discussion above we know 
$(A^{(h_{k-1}-1)}, $ $B^{(h_{k-1}-1)})\in \CF_{h_{k-1}-1}$. 
However, after one $\CU$ update,
$|\gamma(A^{(h_{k-1})}, B^{(h_{k-1})})| \ge \Gamma_{h_{k-1}}$
constitutes a violation of the condition,
implying $(A^{(h_{k-1})}, B^{(h_{k-1})}) \in \CH_{h_{k-1}} \backslash \CF_{h_{k-1}} \subseteq \CG_{h_{k-1}}$.
Furthermore, as all descendants of elements in $\CG_{h_{k-1}}$ lie in $\CG_i$ for some $i$,
we have $(A^{(k-1)}, B^{(k-1)}) \in \CG_{k-1}$, and thus 
$(A^{(k)}, B^{(k)}) \in \CV(A^{(k-1)}, B^{(k-1)})$.
\end{proof}

Therefore, after substituting the corresponding criteria in Formula (\ref{eq:update_prob}), we have
\begin{align}
& ~~~~ \Pr[S_{k} = p_{k} | (S_1, \cdots, S_{k-1}) = (p_1, \cdots, p_{k-1})]  \notag \\
&= \left\{\begin{array}{ll}
\Pr[X = p_{k} - p_{k-1}], & \text{ if } (A^{(k)}, B^{(k)}) \in \CU(A^{(k-1)}, B^{(k-1)}), 
\gamma(A^{(k-1)}, B^{(k-1)}) \ge 0, \\
\Pr[X' = p_{k} - p_{k-1}], & \text{ if } (A^{(k)}, B^{(k)}) \in \CU(A^{(k-1)}, B^{(k-1)}), 
\gamma(A^{(k-1)}, B^{(k-1)}) < 0,  \\
\Pr[Y = p_{k} - p_{k-1}], & \text{ if } (A^{(k)}, B^{(k)}) \in \CV(A^{(k-1)}, B^{(k-1)}), 
\gamma(A^{(k-1)}, B^{(k-1)}) \ge 0, \\
\Pr[Y' = p_{k} - p_{k-1}], & \text{ if } (A^{(k)}, B^{(k)}) \in \CV(A^{(k-1)}, B^{(k-1)}), 
\gamma(A^{(k-1)}, B^{(k-1)}) <0. 
\end{array}\right. \label{eq:Form2}
\end{align}

In light of the fact that for
$(A^{(k)}, B^{(k)}) = (A^{(k-1)} \otimes C^{(k)}, B^{(k-1)} \otimes D^{(k)})$,
$p_k - p_{k-1} = (r_k + \ln(\nnz(C^{(k)})), r_k + \ln(\nnz(D^{(k)}))$, one can perform a case-by-case analysis and obtain the following equivalent form.
\begin{align}
& ~~~~ \Pr[S_{k} = p_{k} | (S_1, \cdots, S_{k-1}) = (p_1, \cdots, p_{k-1})]  \notag \\
&= \left\{\begin{array}{ll}
1/J, & \text{ if } (A^{(k)}, B^{(k)}) = (A^{(k-1)} \otimes U_j, B^{(k-1)} \otimes V_j),
j \in [J],\\
1/J, & \text{ if } (A^{(k)}, B^{(k)} = (A^{(k-1)} \otimes V_j^T, B^{(k-1)} \otimes U_j^T),
j \in [J],\\
1, & \text{ if } (A^{(k)}, B^{(k)}) = (A^{(k-1)} \otimes I, B^{(k-1)} \otimes M), \\
1, & \text{ if } (A^{(k)}, B^{(k)}) = (A^{(k-1)} \otimes M, B^{(k-1)} \otimes I),
\end{array}\right.  \label{eq:Form3} \\
&= \exp(-r_k) \notag
\end{align}

With this result, we compute the conditional expectation in concern by induction on $k = n, n-1, \cdots, 0$.

Base case:
$\Ex[\exp(S_n - S_n) | (S_1, \cdots, S_n) = (S_1, \cdots, S_n)] = 1$.
Meanwhile,
$\sigma(A^{(n)}, B^{(n)}) = (\nnz(I_1),$ $\nnz(I_1)) = (1, 1)$.

Inductive step:  Assume $\Ex[\exp(S_n - S_{k}) | (S_1, \cdots, S_{k}) = (p_1, \cdots, p_{k})] = \sigma(A^{(k)}, B^{(k)})$. Let $\CD$ denote the set of all children of $(A^{(k-1)}, B^{(k-1)})$ on tree $\CT$. Additionally, we remark that all multiplication operations between $\R^2$ vectors are performed entry-wise. Then,
\begin{align*}
& ~~~~ \Ex[\exp(S_n - S_{k-1}) | (S_1, \cdots, S_{k-1}) = (p_1, \cdots, p_{k-1})] \\
&=\Ex[\exp(S_{k} - S_{k-1}) \cdot \exp(S_n - S_{k})| (S_1, \cdots, S_{k-1}) = (p_1, \cdots, p_{k-1})]\\
&=\sum_{(A^{(k)}, B^{(k)}) \in \CD}\big(
\Pr[S_k = p_k | (S_1, \cdots, S_{k-1}) = (p_1, \cdots, p_{k-1})] \\  
&\hspace{7eM} 
\cdot \exp(p_k - p_{k-1})
\odot \Ex[\exp(S_n - S_{k}) | (S_1, \cdots, S_{k}) = (p_1, \cdots, p_{k})]\big)  \\
&=\sum_{(A^{(k)}, B^{(k))} \in \CD}
\exp(-r_k) \cdot
(\exp(r_k) \cdot \nnz(C^{(k)}), \exp(r_k) \cdot \nnz(D^{(k)})) 
\odot \sigma(A^{(k)}, B^{(k)})
\\
&= \sigma (A^{(k-1)}, B^{(k-1)}).
\end{align*}

In conclusion, we have 
\[\Ex[\exp(S_n - S_{k}) | (S_1, \cdots, S_{k}) = (p_1, \cdots, p_{k})] = 
\sigma(A^{(k)}, B^{(k)})\]
for $k = 0, 1, \cdots, n$, as is stated in the proposition.
\end{proof}

At this point, Lemma \ref{lem:SimpleRP} can in fact be proved without additional effort.
Consider Proposition \ref{prop:SimpleRP} with $k = 0$, then we have 
\[\Ex[\exp(S_n - S_0)] = \Ex[\exp(S_n)] = \sigma(A^{(0)}, B^{(0)}) = 
(\nnz(\CC(M^{\otimes n})_{A}), \nnz(\CC(M^{\otimes n})_{B})).\]

\subsection{Another equivalent random walk}
Randomized algorithm (I), despite being closely connected to $\size(\CC(M^{\otimes n}))$, turns out to pose difficulties for further analysis. This issue motivates us to devise a modified randomized algorithm (II) which we will present momentarily. For the sake of conciseness, we will use the same notations as above, though they tend to have different interpretations.
\begin{center}
\fbox{\parbox{0.8\textwidth}{
\textbf{Stage 0}: \\
Draw i.i.d. examples $X_k \sim X, X_k' \sim X', Y_k \sim Y, Y_k' \sim Y', k = 1, \cdots, n$.

\noindent \textbf{Stage 1}:

For $k = 1, 2, \cdots, n$:

~~1. If $T_{k-1} \ge 0$, $S_k = S_{k-1} + X_k$. Otherwise, $S_k = S_{k-1} + X_k'$.

~~2. If $|T_k| \ge \Gamma_k$ (definition of $\Gamma_k$ as before), end Stage 1.

\noindent \textbf{Stage 2}:
Let $k = K$ be the last iteration in Stage 1. 

For $k = K+1, K+2, \cdots, n$:

~~If $T_{k-1} \ge 0$, $S_k = S_{k-1} + Y_k$. Otherwise, $S_k = S_{k-1} + Y_k'$.
}}
(II)
\end{center}

Randomized algorithm (II) shares exactly the same structure with algorithm (I), and differs solely in
the distribution of $X$ and $X'$.
Let \[\pi_j = \frac{\sqrt{\nnz(U_j) \cdot \nnz(V_j)}}{\sum_{j=1}^J \sqrt{\nnz(U_j) \cdot \nnz(V_j)}}.\]
$X, X'$ respectively follow probability distributions
\[\Pr[X = (\ln(1/\pi_j) + \ln \nnz(U_j), 
\ln(1/\pi_j) + \ln \nnz(V_j))] = \pi_j, j \in [J],\]
\[\Pr[X' = (\ln(1/\pi_j) + \ln \nnz(V_j^T), 
\ln(1/\pi_j) + \ln \nnz(U_j^T))] = \pi_j, j \in [J].\]

\begin{proposition}
Let $(A^{(0)}, B^{(0)}), (A^{(1)}, B^{(1)}), \cdots, (A^{(n)}, B^{(n)})$
be a path from root to leaf in \CT. 
$(A^{(i)},$ $B^{(i)})$ $= (A^{(i-1)} \otimes C^{(i)}, B^{(i-1)} \otimes D^{(i)}), i \in [n]$. For $i \in [n]$, define
\[r_i = \left\{\begin{array}{ll}
\ln (1/\pi_j), & \text{ if } (A^{(k)}, B^{(k)}) = (A^{(k-1)} \otimes U_j, B^{(k-1)} \otimes V_j),
j \in [J], \\
\ln (1/\pi_j), & \text{ if } (A^{(k)}, B^{(k)}) = (A^{(k-1)} \otimes V_j^T, B^{(k-1)} \otimes U_j^T),
j \in [J], \\
0, & \text{ if } (A^{(k)}, B^{(k)}) = (A^{(k-1)} \otimes I, B^{(k-1)} \otimes M), \\
0, & \text{ if } (A^{(k)}, B^{(k)}) = (A^{(k-1)} \otimes M, B^{(k-1)} \otimes I),
\end{array}\right. \quad R_i = \sum_{l=1}^{i} r_l.\]
Pick $p_i = (R_i + \ln \nnz(A^{(i)}), R_i + \ln \nnz(B^{(i)})), i \in [n]$. 
Then
\[\Ex[\exp(S_n - S_{k}) | (S_1, \cdots, S_{k}) = (p_1, \cdots, p_{k})] = 
\sigma(A^{(k)}, B^{(k)})\]
holds for any integer $0 \le k \le n$.
\end{proposition}
\begin{proof}
The proof we present here will significantly resemble the one for Proposition \ref{prop:SimpleRP}.
First off, the derivation for $\Pr[S_{k} = s_{k} | (S_1, \cdots, S_{k-1}) = (s_1, \cdots, s_{k-1})]$
does not depend on the specific distribution of $X$ and $X'$. Therefore Formula (\ref{eq:update_prob}) still holds for $X, X', Y, Y'$ in algorithm (II). Similarly, as the derivation for Formula (\ref{eq:Form2}) merely depends on $t_i = \gamma(A^{(i)}, B^{(i)})$, which stays the same here, we know it holds as well.

Further, for
$(A^{(k)}, B^{(k)}) = (A^{(k-1)} \otimes C^{(k)}, B^{(k-1)} \otimes D^{(k)})$,
$p_k - p_{k-1} = (r_k + \ln(\nnz(C^{(k)})), r_k + \ln(\nnz(D^{(k)}))$ is still the case.
Then the same case-by-case analysis shows
\begin{align*}
&~~~~ \Pr[S_{k} = p_{k} | (S_1, \cdots, S_{k-1}) = (p_1, \cdots, p_{k-1})]  \\
&= \left\{\begin{array}{ll}
\pi_j, & \text{ if } (A^{(k)}, B^{(k)}) = (A^{(k-1)} \otimes U_j, B^{(k-1)} \otimes V_j),
j \in [J], \\
\pi_j, & \text{ if } (A^{(k)}, B^{(k)}) = (A^{(k-1)} \otimes V_j^T, B^{(k-1)} \otimes U_j^T),
j \in [J], \\
1, & \text{ if } (A^{(k)}, B^{(k)}) = (A^{(k-1)} \otimes I, B^{(k-1)} \otimes M), \\
1, & \text{ if } (A^{(k)}, B^{(k)}) = (A^{(k-1)} \otimes M, B^{(k-1)} \otimes I),
\end{array}\right. \\
&= \exp(-r_k)
\end{align*}
Namely, $\Pr[S_k = p_k | (S_1, \cdots, S_{k-1}) = (p_1, \cdots, p_{k-1})]$ and
$\exp(p_k - p_{k-1})$ have the same form as before. Hence by induction we are able to draw precisely the same conclusion.
\end{proof}
\begin{corollary}
Let $S_k, k=0, 1, \cdots, n$ be the random vectors in randomized algorithm (II). Then
\[\Ex[\exp(S_n)] = (\nnz(\CC(M^{\otimes n})_{A}), \nnz(\CC(M^{\otimes n})_{B})).\]
\end{corollary}

\subsection{Analysis for the expectation}
\label{subsec:anaexp}
To this point, we have shown the equivalence between the size of $\CC(M^{\otimes n})$ and a statistical characteristic of $S_n$ computed in randomized algorithm (II). This new quantity, unlike its counterpart associated with algorithm (I), allows systematic analysis with the aid of concentration bounds. In the following discussion, all variables, unless otherwise specified, are associated with randomized algorithm (II), as opposed to algorithm (I).

Recall that our primary task here is to give bounds on $\Ex[\exp(S_n)]$. To achieve this goal, we examine two quantities: $|T_n| = |S_n[1] - S_n[2]|$ and $S_n[1] + S_n[2]$. Define $H = h((S_1, \cdots, S_n)) = \min\{i \in [n] ~|~ |T_i| \ge \Gamma_i\}$.

\begin{claim}
$|T_n| = |S_n[1] - S_n[2]| < 4G$.
\end{claim}
\begin{proof}
If $H = h((S_1, \cdots, S_n)) = \infty$, i.e. all iterations of the algorithm stay in Stage 1,
then $|T_n| < \Gamma_n = (n-n) \cdot \gamma(M, I) + 2G < 4G$.

Otherwise, $H$ is the last iteration in Stage 1, which implies $|T_{H-1}| < \Gamma_{H-1},
|T_{H}| \ge \Gamma_{H}$.
Hence we simultaneously have $|T_H| \ge (n-H) \cdot \gamma(M, I) + 2G$ and
\[|T_H| = |T_{H-1}| + (|T_{H}| - |T_{H-1}|) 
< (n-H+1) \cdot \gamma(M, I) + 2G + G 
\le (n-H) \cdot \gamma(M, I) + 4G.\]
Further, for $k = 1, \cdots, n-H$, $T_{H+k-1} \ge (n-k+1) \cdot \gamma(M, I) \ge 0$
implies $T_{H+k} = T_{H+k-1} + Y_{H+k}[1] - Y_{H+k}[2] = (n-k) \cdot \gamma(M, I) \ge 0$.
Therefore by induction, for the case $T_{H} \ge (n-H) \cdot \gamma(M, I) + 2G$,
$T_{n} = T_{H} - (n-H) \cdot \gamma(M, I) < 4G$. Similarly,
if $T_{H} \le -((n-H) \cdot \gamma(M, I) + 2G)$, then $T_{n} > -4G$.
\end{proof}

\begin{claim}
\label{claim:sizesum}
Let 
\[\alpha_1 = \ln (\sum_{j=1}^J \sqrt{\nnz(U_j) \cdot \nnz(V_j)}), \quad
\alpha_2 = \ln (\sqrt{\nnz(M) \cdot \nnz(I)}).\]
Then
\[S_n[1] + S_n[2] = \left\{\begin{array}{ll}
H \cdot 2 \alpha_1 + (n - H) \cdot 2 \alpha_2, & \text{ if } H < \infty, \\
n \cdot 2 \alpha_1, & \text{ if } H = \infty.
\end{array} \right. \]
\end{claim}
\begin{proof}
First we make the observation that
\begin{align*}
&~~~~ X[1] + X[2] \\
&= \ln(1/ \pi_j) + \ln(\nnz(U_j)) + \ln(1 / \pi_j) + \ln(\nnz(V_j)) \\
&= 2 \cdot \left(\ln (\sum_{j=1}^J \sqrt{\nnz(U_j) \cdot \nnz(V_j)}) 
-\frac 12 \ln (\nnz(U_j) \cdot \nnz(V_j)) \right)
+ \ln(\nnz(U_j)) + \ln(\nnz(V_j)) \\
&= 2 \alpha_1.
\end{align*}
Similarly one can show $X'[1] + X'[2] = 2 \alpha_1 = X[1] + X[2]$, as well as $Y[1] + Y[2] = Y'[1] + Y'[2] = \ln(\nnz(M)) + \ln(\nnz(I)) = 2 \alpha_2$.

The results above indicate that $(S_k[1] + S_k[2]) - (S_{k-1}[1] + S_{k-1}[2])$
depends only on whether the $k$-th iteration is in Stage 1 or Stage 2.
In addition, when $H < \infty$, Stage 1 contains $H$ iterations, and Stage 2 contains $n-H$ many;
otherwise $H = \infty$, meaning all the iterations belong to Stage 1.
Thereby the claim to prove immediately follows from the calculations above.
\end{proof}

Combining the two claims, we have
$S_n[1] + S_n[2] \ge S_n[1] + S_n[1] - 4G$, i.e. $S_n[1] \le (S_n[1] + S_n[2])/2 + 2G$.
And similarly, $S_n[2] \le (S_n[1] + S_n[2])/2 + 2G$. That is to say, $S_n[1]$ and $S_n[2]$
are controlled by the quantity in Claim \ref{claim:sizesum}, and thus only dependent on $H$.
Hence it only remains to compute the probability $\Pr[H = k], k = 1, \cdots, n$ and $\Pr[H = \infty]$.

To this end, we first define $Z = X[1] - X[2]$,
whose probability distribution can be obtained as 
$\Pr[Z = \gamma(U_j, V_j)] = \pi_j,  j\in [J]$.
Recall that $G = \max\{\gamma(M, I), \max_{j \in [J]} |\gamma(U_j, V_j)|\}$,
then evidently $|Z| \le G$.
Similarly define $Z' = X'[1] - X'[2]$. One can verify that $-Z'$ has the same distribution as $Z$. 

With this observation, we define random variables 
$Z_k = X_k[1] - X_k[2], Z_k' = X'_k[1] - X'_k[2], k \in [n]$; $E_k = |T_k| - |T_{k-1}|, k \in [n]$.
In Stage 1, as either $T_{k} = (S_{k-1} + X_k)[1] - (S_{k-1} + X_k)[2] = T_{k-1} + Z_k$,
or $T_{k} = (S_{k-1} + X'_k)[1] - (S_{k-1} + X'_k)[2] = T_{k-1} + Z_k'$,
we know $|E_k| \le G$ by triangle inequalities.
In the meantime, from the rules of updating, if $T_{k-1} \ge G$, 
$T_{k} = T_{k-1} + Z_k \ge G - G = 0$ and thus $E_k = T_k - T_{k-1} = Z_k$.
Similarly, if $T_{k-1} \le -G$, $E_{k} = (-T_{k}) - (-T_{k-1}) = -Z_k'$.

Further define $F = \max\{i \in \{0, \cdots, H-1\} ~|~ |T_i| < G\}$.
Clearly $\{i \in \{0, \cdots, H-1\} ~|~ |T_i| < G\} \ne \varnothing$ as $|T_0| = 0$.
Then,
\[\Pr[H = k] \le \Pr[|T_k| = \sum_{i=1}^n E_i \ge \Gamma_k]
= \sum_{f=0}^{k-1} \Pr[F = f] \cdot \Pr[\sum_{i=1}^n E_i \ge \Gamma_k | F = f].\]
For fixed $F = f$, we have $|T_f| = \sum_{i=0}^f E_i < G$ and $E_{f+1} < G$. 
In addition, as $|T_{i}| \ge G$ for $i = f+1, \cdots, k-1$, we know $E_{i}, i = f+2, \cdots, k$ are 
essentially independent copies of $Z$. 

Therefore, let $\TE_i = \frac{E_i + G}{2G} \in [0, 1]$,
$\mu = \Ex[\TE_i] = \frac{\Ex[Z] + G}{2G}$, $m = k - f - 1$, and we have
\[\Pr[\sum_{i = 1}^k E_i \ge \Gamma_k | F = f]
\le \Pr[\sum_{i = f+2}^k E_i \ge \Gamma_k - 2G]
= \Pr[\sum_{i = f+2}^k \TE_i \ge \frac{\Gamma_k - 2G + Gm}{2G}]. \tag{$*$}\]
By Chernoff bound (Theorem \ref{thm:Chernoff}), let
\[\delta = \frac{(\Gamma_k - 2G + Gm)/(2G)}{m \cdot (\Ex[Z] + G)/(2G)} - 1
= \frac{\Gamma_k - 2G}{\Ex[Z] + G} \cdot \frac 1m + \frac{-\Ex[Z]}{\Ex[Z] + G}
=: \frac Am + B, \]
then
\[(*) \le \max\{\exp(-\frac \mu 3 \cdot \delta m), \exp(-\frac \mu 3 \cdot \delta^2 m)\}
= \exp(-\frac \mu 3 \cdot \min\{\delta m, \delta^2 m\}),\]
\[\delta m = A + Bm \ge A, \quad 
\delta^2 m = \frac{A^2}m + 2AB + B^2m \ge 4AB.\]
Therefore,
\begin{align*}
\Pr[H = k] &\le \max_f \Pr[\sum_{i = 1}^k E_i \ge \Gamma_k | F = f]  \\
&\le \exp(-\frac {\gamma(M \times I)}{6G} \cdot \min \{1, \frac{-4\Ex[Z]}{\Ex[Z] + G}\} \cdot (n-k))
=: \exp(-\beta(n-k)).
\end{align*}
At this point, we have obtained all ingredients required for computing $\Ex[\exp(S_n)]$.
Indeed,
\begin{align*}
&~~~~ \max\{\Ex[\exp(S_n[1])], \Ex[\exp(S_n[2])]\} \\
&\le \exp(2G) \cdot \Ex[\exp(\frac {S_n[1] + S_n[2]}2)] \\
&\le \exp(2G) \cdot \left[
\Pr[H = \infty] \cdot \exp(n \alpha_1)
+ \sum_{k=1}^{n} \exp(-\beta(n-k)) \cdot \exp(k \alpha_1 + (n-k) \alpha_2)
\right] \\
&=\exp(2G) \cdot \exp(n\alpha_1) \cdot \left[
\Pr[H = \infty]
+ \sum_{k=1}^{n} \exp((n-k) (\alpha_2 - \alpha_1 - \beta))
\right] 
\end{align*} 
Under the condition $\beta > \alpha_2 - \alpha_1$, we have
$\sum_{k=1}^{n} \exp((n-k) (\alpha_2 - \alpha_1 - \beta)) \le n$.
Therefore, 
\[\size(\CC(M^{\otimes n}))
\le 2 \cdot \max\{\Ex[\exp(S_n[1])], \Ex[\exp(S_n[2])]\}
\le \exp(\alpha_1)^{n + o(n)}.\]

\subsection{One-sided decomposition and rigidity-based decomposition}
Theorem \ref{thm:main} shows that under a certain condition, $M^{\otimes n}$ has a depth-2 linear circuit of size $\exp(\alpha_1)^{n+o(n)}$. In this Subsection, however, we will be concerned with a special type of decomposition, for which this condition can be removed.

\begin{definition}
Suppose $q, J$ are positive integers, and $M \in \F^{q \times q}$ is a symmetric matrix. 
$M = \sum_{j=1}^J U_j \times V_j$ is called a one-sided decomposition if $\nnz(U_j) \le \nnz(V_j)$
for all $j \in [J]$.
\end{definition}
For this type of decomposition, we have the following lemma.
\begin{lemma}
If $M = \sum_{j=1}^J U_j \times V_j$ is a one-sided decomposition, then $\Pr[H = \infty] = 1$ ($H$ is defined as in Subsection \ref{subsec:anaexp}). 
\end{lemma}
\begin{proof}
Consider the random variable $Z$. From the discussion above we know $|Z| \le G$.
Now as $\gamma(U_j, V_j) = \ln(\nnz(U_j) / \nnz(V_j)) \le 0$, we further have
$-G \le Z \le 0$.
This implies that if $0 \le T_{k-1} \le G$, $-G \le T_k = T_{k-1} + Z_k \le G$, 
and similarly if $-G \le T_{k-1} < 0$, $-G \le T_k = T_{k-1} + Z_k' \le G$.
Therefore, we know $-G \le T_k \le G$ holds for all $k \in [n]$,
since initially $T_0 = 0$.

In other words, $|T_k| \le G \le \Gamma_k = (n-k) \cdot \gamma(M, I) + 2G$
holds for all $k \in [n]$, meaning the algorithm passes all the tests
and stays in Stage 1 eventually. By definition, $H = \infty$ is always the case.
\end{proof}
Based on this lemma, one can carry out the same calculation and obtain
\[\exp(2G) \cdot \Ex\left[\exp\left(\frac {S_n[1] + S_n[2]}2\right)\right]
= \exp(2G) \cdot \Pr[H = \infty] \cdot \exp(n \alpha_1)
= O(\exp(\alpha_1)^n)\]
\begin{theorem}
\label{thm:onesided}
Suppose $q, J$ are positive integers, and $M \in \F^{q \times q}$ is a symmetric matrix, which has one-sided decomposition
$M = \sum_{j = 1}^J U_j \times V_j$. Let
\[\alpha_1 = \ln \left(\sum_{j=1}^J \sqrt{\nnz(U_j) \cdot \nnz(V_j)}\right),\]
then for any positive integer $n$, $M^{\otimes n}$ has a depth-2 linear circuit of size
$O(\exp(\alpha_1)^n)$.
\end{theorem}

Our motivation for studying this special case primarily stems from the fact
that rigidity upper bounds of a matrix naturally gives rise to a one-sided decomposition for the same matrix.
By definition, $M \in \F^{q \times q}$ has rank-$r$ rigidity $\rig_M(r)$ 
means there exist matrix $L, S \in \F^{q \times q}$ with $\rank(L) \le r, \nnz(S) \le \rig_M(r)$,
such that $M = L + S$. By rewriting, one further has $M = U \times V + I \times S$
for some $U \in \F^{q \times r}, V \in \F^{r \times q}$.
For the sake of simplicity, we w.l.o.g. assume $\nnz(U) = \nnz(V)$ are precisely $q \cdot r$,
as we will only use the pattern of non-zero entries in $U, V$.
Then we have the following corollary.

\begin{corollary}
\label{cor:rigLC}
Suppose $q, r$ are positive integers, and $M \in \F^{q \times q}$ is a symmetric matrix, with rank-$r$
rigidity $\rig_{M}(r) \ge q$. Then for any positive integer $n$, $M^{\otimes n}$ has a depth-2 linear circuit of size at most
$O\left(\left(q \cdot r + \sqrt{q \cdot \rig_M(r)}\right)^n\right)$.
\end{corollary}
\begin{proof}
Consider the decomposition $M = U \times V + I \times S$ shown above.
One can easily verify the one-sided property as
$\nnz(I) = q \le \nnz(S) = \rig_M(r)$, $\nnz(U) = \nnz(V) = q \cdot r$.
Therefore, directly applying Theorem \ref{thm:onesided} yields the desired result.
\end{proof}
\begin{remark}
\label{rem:asymLC} \label{rem:asym}
It is worth noting that Corollary \ref{cor:rigLC} works for asymmetric matrix $M$ as well. Indeed, one can instead prove Theorem \ref{thm:main} with decompositions $M = U \times V + I \times S$ and $M = U \times V + S \times I$, because $S$ and $I$ naturally commute. In this way the usage of property $M = M^T$ is circumvented,
giving rise to such a linear circuit for arbitrary matrix $M$.
\end{remark}
\begin{remark}
In Theorem \ref{thm:main}, Theorem \ref{thm:onesided} and Corollary \ref{cor:rigLC},
if the base matrix is itself a Kronecker power, i.e., $M = A^{\otimes k}$ for some matrix $A$ and 
fixed positive integer $k$,
we immediately have a small depth-2 linear circuit of matrix $A^{\otimes n}$, when $n$ is a positive multiple of $k$.
Whereas in the case where $k$ does not divide $n$, we instead apply this construction for the next multiple $n' > n$
of $k$, and then pick the appropriate submatrix of $M^{\otimes n'}$. One can verify that the size bound obtained blows up by at most a constant factor.
Therefore, in the following discussion, when the base matrix is indeed picked in this fashion,
the divisibility constraint is always implicitly removed.
\end{remark}

\subsection{Connection between our result and \cite{Alm21}}
\label{subsec:comparison}
One of the main results of~\cite{Alm21} gave a way to construct circuits for Kronecker powers based on rigidity upper bounds. We note in this section that this prior approach could not give a better depth-2 circuit size for $H_n$ than $O(N^{1.476})$, and that our Theorem \ref{thm:onesided} strictly improves the prior approach.
\begin{theorem}[\cite{Alm21}]\label{thm:old}
Suppose $q, n, r$ are positive integers, $N = 2^n$, and $M \in \F^{q \times q}$ is a symmetric matrix, with rank-$r$
rigidity $\rig_{M}(r) \ge q$. Then $M^{\otimes n}$ has a depth-2 linear circuit of size at most \\
$O((\sqrt{(qr + q) \cdot (qr + \rig_M(r))})^n) = O(N^{1 + c/2})$, where $c = \log_N((r+1)(r + \rig_{H_n}(r)/N))$.
\end{theorem}
\begin{proof}
Using the same notation as in the previous subsection, the result directly follows from Theorem \ref{thm:onesided} accompanied with one-sided decomposition
\[M =\left(\begin{array}{c|c} U & I \end{array}\right) \times
\left(\begin{array}{c} V \\ \hline S \end{array}\right).\]
\end{proof}

This bound, by AM-GM inequality is provably worse than the result of Corollary \ref{cor:rigLC}.
Furthermore, as we will prove shortly, $O(N^{1.476})$ given in \cite{Alm21} is the optimal size achievable using Alman's approach, in conjunction with rank-1 rigidity bounds. By comparison, the same non-rigidity results leads to a size bound $O(N^{1.443})$ shown in Theorem \ref{thm:WHLCintro}.

In fact, to construct a smaller circuit using Alman's approach, one must use a rank-1 rigidity upper bound other than $\rig_{H_4}(1) = 96$. However, we verify using brute-force search that $\rig_{H_5}(1) = 432$, which further implies rigidity lower bound
$\rig_{H_n}(1) \ge 432 \cdot 4^{n-5}$ for any integer $n \ge 5$. (This uses the known fact that, for any positive integers $r \leq n$, we have $\rig_{H_n}(r) \geq 4 \cdot \rig_{H_{n-1}}(r)$.)
Then after plugging this result into $c = \log_{2^n}(2 \cdot (1 + \rig_M(r)/2^n))$ from Theorem~\ref{thm:old} and minimizing the result over integers $n$, we find that $c < 0.952$ is indeed the minimum, corresponding to $H_4$ as used by the prior work.

\section{Kronecker power of 2-by-2 matrices and Walsh-Hadamard matrices}
In this section we will first show a rank-1 rigidity upper bound for Kronecker power of 2-by-2 matrices, then for Walsh-Hadamard matrices. 
With these rigidity bounds, we therefore make use of Corollary \ref{cor:rigLC} to construct small linear circuits for both matrices.

\begin{theorem}
\label{thm:rigKron}
Suppose $n$ is a positive integer, and $M \in \F^{2 \times 2}$ is a $2 \times 2$ matrix. Let $M_n = M^{\otimes n}$, then we have
\[\rig_{M_n}(1) \le \left\{\begin{array}{ll}
2^{2n} - 2^n \cdot \binom {n+1} {n/2}
& \text{ if } n \text{ is even}, \\
2^{2n} - 2^{n-1} \cdot \binom {n+2} {(n+1)/2} 
& \text{ if } n \text{ is odd}. \\
\end{array}\right.\]
\end{theorem}
To prove this theorem, it suffices to construct matrices $U \in \F^{2^n \times 1},
V \in \F^{1 \times 2^n}$, such that $L := U \times V$ agrees with $M_n$ on at least 
$2^n \cdot \binom {n+1}{n/2}$ elements for even $n$,
and at least $2^{n-1} \cdot \binom {n+2} {(n+1)/2}$ elements for odd $n$.
Based on Lemma \ref{lem:diag} and Lemma \ref{lem:outer1}, we w.l.o.g. assume $M$ is an outer-1 matrix, i.e. $M$ takes the form $M = \begin{pmatrix} 1 & 1 \\ 1 & \omega \end{pmatrix}$ with
some $\omega \in \F$.
Then the elements of $M_n$ can be explicitly written as $M_n[x, y] = \omega^{\ip x y}$ for $x, y \in\{0, 1\}^n$.

Before entering the formal proof, we present here some intuitions behind the construction.
Denote $N = 2^n$, and $\bo \in \F^{N}$ to be the all-1 vector. 
We first make the observation that
$\ip{x - \frac 12 \bo}{y - \frac 12 \bo} =
\ip {x}{y} - \frac 12\ip{x}{\bo} - \frac 12\ip{\bo}{y} + \frac n4$.
This means we can form diagonal matrices $D_1, D_2$ with 
$D_1[x, x] = \omega^{- \frac 12\ip{x}{\bo}}$ and 
$D_2[y, y] = \omega^{- \frac 12\ip{\bo}{y} + \frac n4}$ such that 
$M_n' = D_1 \times M_n \times D_2$ is precisely the matrix with 
elements $M_n'[x, y] = \omega^{\ip{x - \frac 12 \bo}{y - \frac 12 \bo}}$.
Hence again by Lemma \ref{lem:diag}, we can refocus on matrix $M_n'$, which has the same rigidity as $M_n$.

At this point, if we denote $c$ to be the most frequent value of
$\ip{x - \frac 12 \bo}{y - \frac 12 \bo}$ 
over all pairs $(x, y) \in \{0, 1\}^n \times \{0, 1\}^n$, 
a constant matrix filled with elements $\omega^{c}$ is already 
a promising rank-1 approximation to $M_n'$.
However, following this line, we can construct a better rank-1 matrix
by noticing that the frequency of $\ip{x - \frac 12 \bo}{y - \frac 12 \bo}$ varies among blocks corresponding to different parity of $|x| = \ip{x}{\bo}$ and $|y| = \ip{\bo}{y}$.
For example, when $n$ is a multiple of 4, the most frequent value of $\ip{x - \frac 12 \bo}{y - \frac 12 \bo}$ is $0$ in the block with both $|x|$ and $|y|$ odd, and $\pm 1/2$ in the block with $|x|$ odd and $|y|$ even.
Therefore, if we partition $M_n'$ with respect to the parity conditions mentioned above,
and find the most frequent values in each block, then
matrix $L'$ consisting of four constant blocks filled with corresponding powers of $\omega$
is expected to be a better rank-1 approximation.

It is also worth remarking that the intuitions stated above are rather informal. Most technical issues,
including $- \frac 12\ip{x}{\bo}$ is not necessarily an integer, the matrix composed of constant blocks may have rank higher than 1, etc., are left to be solved in the formal proof.
Additionally, in order to integrate cases with different $\omega$, we will make use of $M_n'$ only implicitly in the formal proof.

\subsection{Proof of theorem \ref{thm:rigKron}}
\paragraph{Case 1: $n \equiv 0 \pmod 4$.}
Let $b_1, b_2 \in \Z^{N}$ be vectors with 
\[b_1[x] = \floor{\frac 12\ip{x}{\bo}}, 
x \in \{0, 1\}^n, \quad
b_2[y] = \ceil{\frac 12\ip{\bo}{y}} - \frac n4,
y \in \{0, 1\}^n.\]
Then construct matrices $U \in \F^{N \times 1}, V \in \F^{1 \times N}$ such that for $x, y \in \{0, 1\}^n, U[x, 1] = \omega^{b_1[x]}, V[1, y] = \omega^{b_2[y]}$.
One can thus see that if $b_1[x] + b_2[y] = \ip{x}{y}$, then $L[x, y] = U[x, 1] \cdot V[1, y] = M_n[x, y]$. Therefore, $\nnz(M_n - L) \le N^2 - \#[b_1[x] + b_2[y] = \ip{x}{y}]$.

By definition, $\ip{x}{y} - (b_1[x] + b_2[y]) = 
\ip {x}{y} - \floor{\frac 12\ip{x}{\bo}} - \ceil{\frac 12\ip{\bo}{y}} + \frac n4$.
We can rewrite $\floor{\frac 12\ip{x}{\bo}}$ and $\ceil{\frac 12\ip{\bo}{y}}$ as
\[\floor{\frac 12\ip{x}{\bo}} = \frac 12\ip{x}{\bo} + 
\left\{\begin{array}{ll}
-1/2, & \text{ if } \ip {x}{\bo} \text{ is odd}, \\
0, & \text{ if } \ip {x}{\bo} \text{ is even}.
\end{array}\right. ~~
\ceil{\frac 12\ip{\bo}{y}} = \frac 12\ip{\bo}{y} + 
\left\{\begin{array}{ll}
1/2, & \text{ if } \ip{\bo}{y} \text{ is odd}, \\
0, & \text{ if } \ip{\bo}{y} \text{ is even}.
\end{array}\right.\]
Therefore, if we define $F \in \Q^{N \times N}$ to be the matrix such that for $x, y\in \{0, 1\}^n$,
\[F[x, y] = \left\{\begin{array}{ll}
-1/2, & \text{ if } \ip {x}{\bo} \text{ is odd}, \ip{\bo}{y} \text{ is even}, \\
0, & \text{ if } \ip {x}{\bo}, \ip{\bo}{y} \text{ are both odd or both even},\\
1/2, & \text{ if } \ip {x}{\bo} \text{ is even}, \ip{\bo}{y} \text{ is odd},
\end{array}\right.\]
then $b_1[x] + b_2[y] = \ip{x}{y}$ is equivalent to 
\[\ip{x}{y} - (b_1[x] + b_2[y]) = \ip {x}{y} - \frac 12\ip{x}{\bo}
- \frac 12\ip{\bo}{y} + \frac n4 - F[x, y] 
= \ip{x - \frac 12}{y - \frac 12} - F[x, y] = 0.\]
Here $F$ is in fact the matrix with the most frequent values of 
$\ip{x - \frac 12}{y - \frac 12}$ in respective blocks.

Now it only remains to count the number of $(x, y)$ pairs satisfying the condition above.
For simplicity, we use the term \emph{good pair} to denote such a pair in the following discussion.
For a good pair $(x, y)$, we partition $[n]$ into 2 sets as follows. 
For $a \in \{0, 1\}$, define set $I_{a} = \{i \in [n] ~|~ x[i] \oplus y[i] = a\} \subseteq [n]$.
Then we simultaneously have
$|I_0| + |I_1| = n$, and
$\ip{x - \frac 12 \bo}{y - \frac 12 \bo} - F[x, y] = \frac 14 (|I_0| - |I_1|) - F[x, y] = 0$,
which implies
\[|I_0| = \frac n2 + 2F[x, y], \quad 
|I_1| = \frac n2 - 2F[x, y].\]

Therefore we have the following procedure which precisely generates all good pairs. 
\begin{enumerate}
\item Pick the parity of $\ip{x}{\bo}$ and $\ip{\bo}{y}$.
\item Pick a subset of $[n]$ of size $(n/2 + 2F[x, y])$ to be $I_0$ 
($F[x, y]$ is determined by the parity combination); $I_1 = [n] \backslash I_0$.
\item For each $i \in I_0 \cap [n-1]$, pick $(x[i], y[i])$ to be either $(0, 0)$ or $(1, 1)$.
\item For each $i \in I_1 \cap [n-1]$, pick $(x[i], y[i])$ to be either $(0, 1)$ or $(1, 0)$.
\item Use $(x[n], y[n])$ to adjust the parity of $\ip{x}{\bo}$ and $\ip{\bo}{y}$,
such that they have the same parity as what is picked in step 1.
\end{enumerate}

Define matrix $\TF \in \Q^{\{0, 1\} \times \{0, 1\}}$ to be the representatives corresponding to the four blocks of $F$: For $i, j \in \{0, 1\}$, $\TF[i, j] = F[x, y]$ if $\ip {x}{\bo} \equiv i \pmod 2, \ip{\bo}{y} \equiv j \pmod 2$ . (This is well-defined by the definition of $F$.)
Then following the procedure for generating good pairs, we can compute the number of good pairs 
$g:= \#[\ip{x - \frac 12 \bo}{y - \frac 12 \bo} - F[x, y] = 0]$ as
\begin{align*}
g &= \sum_{i, j \in \{0, 1\}} \binom{n}{n/2 - 2\TF[i, j]} 
\cdot 2^{n/2 - 2\TF[i, j]} \cdot 2^{n/2 + 2\TF[i, j]} \cdot 2^{-1} \\
&= \left[\binom{n}{n/2-1} + 2\binom{n}{n/2} + \binom{n}{n/2+1}\right] \cdot 2^{n-1} 
= 2^{n} \cdot \binom{n+1}{n/2}.
\end{align*}
Hence we have $\rig_{M_n}(1) \le 
\nnz(M - L) \le N^2 - \#[b_1[x] + b_2[y] = \ip{x}{y}] = 2^{2n} - 2^{n} \cdot \binom{n+1}{n/2}$ as desired.

For other cases, we only compute $g$, from which the claim immediately follows.
\paragraph{Case 2: $n \equiv 2 \pmod 4$.} Similarly, consider vectors $b_1, b_2 \in \Z^{N}$ with
\[b_1[x] = \ceil{\frac 12\ip{x}{\bo}},
x \in \{0, 1\}^n, \quad
b_2[y] = \ceil{\frac 12\ip{\bo}{y}} - \frac{n+2}4,
y \in \{0, 1\}^n.\]
The same calculation shows that for $x, y \in \{0, 1\}^n$, $b_1[x] + b_2[y] = \ip{x}{y}$ is equivalent to 
$\ip{x - \frac 12}{y - \frac 12} - F[x, y] = 0$.
Here $F[x, y]$ is defined as
\[H[x, y] = \left\{\begin{array}{ll}
-1/2, & \text{ if } \ip {x}{\bo}, \ip{\bo}{y} \text{ are both even}, \\
0, & \text{ if } \ip {x}{\bo}, \ip{\bo}{y} \text{ have different parity},\\
1/2, & \text{ if } \ip {x}{\bo}, \ip{\bo}{y} \text{ are both odd}.
\end{array}\right.\]
Therefore we have
\[g= \left[\binom{n}{n/2-1} + 2\binom{n}{n/2} + \binom{n}{n/2+1}\right] \cdot 2^{n-1} 
= 2^{n} \cdot \binom{n+1}{n/2}.\]

\paragraph{Case 3: $n \equiv 1 \pmod 4$.} Consider vectors $b_1, b_2 \in \Z^N$ with
\[b_1[x] = \floor{\frac 12\ip{x}{\bo}},
x \in \{0, 1\}^n, \quad
b_2[y] = \ceil{\frac 12\ip{\bo}{y}} - \frac{n-1}4,
y \in \{0, 1\}^n.\]
For $x, y \in \{0, 1\}^n$, 
$b_1[x] + b_2[y] = \ip{x}{y} \Leftrightarrow \ip{x - \frac 12}{y - \frac 12} - F[x, y] = 0$,
where
\[F[x, y] = \left\{\begin{array}{ll}
-1/4, & \text{ if } \ip {x}{\bo} \text{ is odd}, \ip{\bo}{y} \text{ is even}, \\
1/4, & \text{ if } \ip {x}{\bo}, \ip{\bo}{y} \text{ are both odd or both even},\\
3/4, & \text{ if } \ip {x}{\bo} \text{ is even}, \ip{\bo}{y} \text{ is odd},
\end{array}\right.\]
\[g = \left[\binom{n}{(n-1)/2} + 2\binom{n}{(n+1)/2} + \binom{n}{(n+3)/2}\right] \cdot 2^{n-1}
= 2^{n-1} \cdot \binom{n+2}{(n+1)/2}.\]

\paragraph{Case 4: $n \equiv 3 \pmod 4$.} Consider vectors $b_1, b_2 \in \Z^N$ with
\[b_1[x] = \ceil{\frac 12\ip{x}{\bo}},
x \in \{0, 1\}^n, \quad
b_2[y] = \ceil{\frac 12\ip{\bo}{y}} - \frac{n+1}4,
y \in \{0, 1\}^n.\]
For $x, y \in \{0, 1\}^n$, 
$b_1[x] + b_2[y] = \ip{x}{y} \Leftrightarrow \ip{x - \frac 12}{y - \frac 12} - F[x, y] = 0$,
where
\[F[x, y] = \left\{\begin{array}{ll}
-3/4, & \text{ if } \ip {x}{\bo}, \ip{\bo}{y} \text{ are both even}, \\
-1/4, & \text{ if } \ip {x}{\bo}, \ip{\bo}{y} \text{ have different parity},\\
1/4, & \text{ if } \ip {x}{\bo}, \ip{\bo}{y} \text{ are both odd}.
\end{array}\right.\]
\[g = \left[\binom{n}{(n-3)/2} + 2\binom{n}{(n-1)/2} + \binom{n}{(n+1)/2}\right] \cdot 2^{n-1}
= 2^{n-1} \cdot \binom{n+2}{(n+1)/2}.\]
\begin{remark}
In Case 3 and 4, not all elements in $\TF$ can be picked as the most frequent value in the corresponding block because of the rank constraint. As one can see
$F[x, y] = 3/4$ and $F[x, y] = -3/4$ are in fact the second frequent values in the pertinent block.
\end{remark}

\subsection{Rigidity of Walsh-Hadamard matrices}
Walsh-Hadamard matrices can be seen as a special case of Kronecker powers with $\omega = -1$.
Therefore, the rigidity upper bound shown in Theorem \ref{thm:rigKron} automatically holds for Walsh-Hadamard matrices as well. However, we can further take advantage of the small multiplicative order of $\omega = -1$, and improve on the bound as follows.

\begin{theorem}
\label{thm:rigWH}
Suppose $n$ is a positive integer, and $H_n$ is the Walsh-Hadamard matrix of size $2^n \times 2^n$. Then we have
\[\rig_{H_n}(1) \le 
\left\{\begin{array}{ll}
2^{2n-1} - 2^{1.5n-1}
& \text{ if } n \text{ is even}, \\
2^{2n-1} - 2^{1.5(n-1)}
& \text{ if } n \text{ is odd}. \\
\end{array}\right.
\]
\end{theorem}

Let $b_1, b_2, U, V, L, F$ be defined as above (respectively for four cases).
From the previous proof, we know that if $b_1[x] + b_2[y] = \ip{x}{y}$, then
$L[x, y] = M_n[x, y]$.
When $\omega = -1$, the condition can be relaxed.
In fact, if $b_1[x] + b_2[y] \equiv \ip{x}{y} \pmod 2$,
then $b_1[x] + b_2[y] - \ip{x}{y}$ is even, and thus
\[L[x, y] = U[x, 1] \cdot V[1, y] = (-1)^{\ip{x}{y}}
\cdot (-1)^{b_1[x] + b_2[y] - \ip{x}{y}} = H_n[x, y].\]
Therefore $\nnz(H_n - L) \le 2^{2n} - \#[b_1[x] + b_2[y] - \ip{x}{y} \equiv 0 \pmod 2]
=2^{2n} - \#[\ip{x - \frac 12 \bo}{y - \frac 12 \bo} - F[x, y] \equiv 0 \pmod 2]$.

Then we can use the a slightly modified procedure to generate all good pairs.
Now after step 1, we first pick integer $k$ such that 
$0 \le n/2 - 2(F[x, y] + 2k) \le n$.
Or equivalently, we pick integer $0 \le t \le n$ such that $t \equiv n/2 - 2F[x, y] \pmod 4$. 
Then we perform step 2 by picking a subset of $[n]$ of size $t$.
Hence the number of good pairs $g:= \#[\ip{x - \frac 12 \bo}{y - \frac 12 \bo} - F[x, y] = 0 \pmod 2]$ can be computed as
\[ g = \sum_{i, j \in \{0, 1\}}
\sum_{\begin{substack}{t \equiv n/2 - 2\TF[i, j] \text{ (mod 4)}\\
0 \le t \le n}\end{substack}}
\binom{n}{t}
\cdot 2^{t} \cdot 2^{n - t} \cdot 2^{-1}.\]
By series multisection identities (Lemma \ref{lem:multisection}),
when $n$ is even,
\[g = 2^{n-1} \cdot 
\left[\MS(n, 4, \frac n2 - 1) + 2 \cdot \MS(n, 4, \frac n2) + \MS(n, 4, \frac n2 +1)\right]
= 2^{2n-1} + 2^{1.5n-1}.\]
When $n$ is odd,
\begin{align*}
g &= 2^{n-1} \cdot 
\left[\MS(n, 4, \frac n2 + \frac 12) + 2 \cdot \MS(n, 4, \frac n2 - \frac 12) 
+ \MS(n, 4, \frac n2 - \frac 32)\right] \\
&= 2^{n-1} \cdot 
\left[\MS(n, 4, \frac n2 - \frac 12) + 2 \cdot \MS(n, 4, \frac n2 + \frac 12) 
+ \MS(n, 4, \frac n2 + \frac 32)\right] \\
&= 2^{2n-1} + 2^{1.5(n-1)}.
\end{align*}
Thereby $\rig_{H_n}(1) \le \nnz(H_n-L)\le 2^{2n} - g$, which is precisely the statement of
Theorem \ref{thm:rigWH}.

\subsection{Linear circuits for Kronecker power and Walsh-Hadamard matrices}
With the rigidity upper bounds shown above, we can apply Corollary \ref{cor:rigLC}
with base matrix respectively $R_k$ and $H_k$ for some positive integer $k$,
then obtain depth-2 linear circuits as below.

\begin{theorem}
Let $M \in \F^{2 \times 2}$ be a $2 \times 2$ matrix on field $\F$.
For any positive integer $k, n$, let $N = 2^n$. Then
$M_n = M^{\otimes n}$ has a depth-2 linear circuit of size
$O(N^{c(k)})$, where
\[c(k) = \frac{\log\left(2^k + 2^k \sqrt{2^{k} - \binom{k+1}{k/2}}\right)}{\log (2^k)}\]
if $k$ is even, while
\[c(k) = \frac{\log\left(2^k + 2^k \sqrt{2^{k} - \frac 12 \binom{k+2}{(k-1)/2}}\right)}{\log (2^k)}\]
if $k$ is odd. In particular, when $k=6$,
$M_n = M^{\otimes n}$ has a depth-2 linear circuit of size $O(N^{c(6)}) = O(N^{1.446})$.
\end{theorem}
\begin{theorem}
For any field $\F$ and any positive integer $k, n$, let $N = 2^n$.
Then the Walsh-Hadamard matrix $H_n \in \F^{N \times N}$ 
has a depth-2 linear circuit of size
$O(N^{c(k)})$, where
\[c(k) = \frac{\log\left(2^k + 2^k \sqrt{2^{k-1} - 2^{(k-2)/2}}\right)}{\log (2^k)}\]
if $k$ is even, while
\[c(k) = \frac{\log\left(2^k + 2^k\sqrt{2^{k-1} - 2^{(k-3)/2}}\right)}{\log (2^k)}\]
if $k$ is odd. In particular, when $k=6$,
$M_n = M^{\otimes n}$ has a depth-2 linear circuit of size $O(N^{c(6)}) = O(N^{1.443})$.
\end{theorem}

\section{General Kronecker power matrices}

\begin{theorem}
\label{thm:rigGKron}
Suppose $q$ is a positive integer, and $M \in \F^{q \times q}$ is a $q \times q$ matrix on field $\F$.
For any real value $0 < s < 1$ and positive integer $n$,
let $M_n =  M^{\otimes n}$, $N = q^n$. Then we have
\[\rig_{M_n}(N^s) \le O(N^{2-h(s, q)}),\]
where $h(s, q) = \Omega(\frac {s^2}{\log^2(2/s)} \cdot \frac 1{q^2 \log q}).$
\end{theorem}

Let $M \in \F^{q \times q}$ be a matrix indexed by $x, y \in [q]$.
Define function $e_s: [q] \rightarrow \{0, 1\}$ as 
$e_s(a) = \bo[a = s]$ for $a \in [q]$.
Then for $a, b \in [q]$ we can rewrite the entries of $M$ as 
\[M[a, b] = \prod_{s, t\in [q]} M[s, t]^{e_s(a) \cdot e_t(b)}\]
Further by the definition of Kronecker product, for $x, y \in [q]^n$,
\[M_n[x, y] = \prod_{i=1}^n M[x[i], y[i]] 
= \prod_{s, t \in [q]} M[s, t]^{\sum_{i=1}^n e_s(x[i]) \cdot e_t(y[i])}
=: \prod_{s, t\in [q]} M[s, t]^{e_{s, t}(x, y)}.\]

Now we use an interpolation polynomial as an approximation to 
function $M[s, t]^z$.
More specifically, by Theorem \ref{thm:interpolation},
for any $s, t \in [q]$, if we define 
$f_{s, t}(z) = M[s, t]^{\sum_{i=1}^n z[i]}$
for $z = (z[1], \cdots, z[n]) \in \{0, 1\}^n$, then there exists a polynomial $\hat f_{s, t}(z)$ of degree at most $(h-l)$ with coefficients in $\F$, such that $\hat f_{s, t}(z) = f_{s, t}(z)$ for any $z \in \{0, 1\}^n$ with $l \le |z| \le h$.
Here we pick integers $l = (\frac 1{q^2} - \delta)\cdot n, h = (\frac 1{q^2} + \delta)\cdot n$ uniformly for all polynomials $\hat f_{s, t}$, where $\delta = \frac 1{cq^2}$ with parameter $c>4$ independent of $q$.

At this point, one can see that if we let $\hat g_{s, t}(x, y) := \hat f_{s, t} (e_s(x[1])\cdot e_t(y[1]), \cdots, e_s(x[n])\cdot e_t(y[n]))$ and 
$\hat g(x, y) := \prod_{s, t\in [q]} \hat g_{s, t}(x, y)$, then $\hat g$
correctly computes a fair number of entries in $M$.
This observation motivates us to make use of the polynomial method shown in Lemma \ref{lem：polymethod}, which further requires computing the following two quantities.

First we count the number of \emph{bad pairs}, i.e. the pairs 
failing to satisfy
$M_n[x, y] = \hat g(x, y)$.
From the definition of $\hat f_{s, t}$, the condition holds if
$l \le e_{s, t}(x, y)\le h$ for all $s, t \in [q]$. 
Hence by union bound,
\begin{align*}
& ~~~~ \#\left[
\exists s, t \in [q] \text{ s.t. } e_{s, t}(x, y) < l \text{ or }
e_{s, t}(x, y) > h
\right] \\
&\le \sum_{s, t \in [q]} 
\left(
\#\left[\sum_{i=1}^n \bo[(x[i], y[i]) = (s, t)] < l \right] + 
\#\left[\sum_{i=1}^n \bo[(x[i], y[i]) = (s, t)] > h \right]
\right) \\
&=q^2 \left(\sum_{k=0}^{l-1} \binom{n}{k} (q^2 - 1)^{n-k}
+ \sum_{k=h+1}^{n} \binom{n}{k} (q^2 - 1)^{n-k} \right) \\
&\le q^2 \cdot O(n) \cdot 
(2^{n \cdot H(l/n) + (n-l) \cdot \log(q^2-1)} +  
2^{n \cdot H(h/n) + (n-h) \cdot \log(q^2-1)}),
\end{align*}
where the exponents can be bounded using Lemma \ref{lem:entropy} as
\begin{align*}
& ~~~~ n \cdot H(\frac ln) + (n-l) \cdot \log(q^2-1) 
=n \cdot [H(\frac 1{q^2} - \delta)
+ (1 - \frac1{q^2} + \delta) \log(q^2-1)] \\
&\le n \cdot [
\log(q^2) - \frac{q^2-1}{q^2} \log(q^2-1)
- \delta \log(q^2 - 1) + H''(\frac1{q^2})\frac{\delta^2}2
+ (1 - \frac1{q^2} + \delta) \log(q^2-1)] \\
&=n \cdot [\log(q^2) + H''(\frac1{q^2})\frac{\delta^2}2], \\
& ~~~~ n \cdot H(\frac hn) + (n-h) \cdot \log(q^2-1)
=n \cdot [H(\frac 1{q^2} + \delta)
+ (1 - \frac1{q^2} - \delta)  \log(q^2-1)] \\
&\le n \cdot [
\log(q^2) - \frac{q^2-1}{q^2} \log(q^2-1)
+ \delta  \log(q^2 - 1) + H''(\frac1{q^2} + \delta)\frac{\delta^2}2
+ (1 - \frac1{q^2} - \delta)  \log(q^2-1)] \\
&=n \cdot [\log(q^2) + H''(\frac{1}{q^2} + \delta)\frac{\delta^2}2].
\end{align*}
Since $\delta = \frac 1{cq^2} < \frac{1}{q^2}$, the two quantities above can be further bounded by 
$n \cdot [\log(q^2) + H''(\frac 2{q^2})\frac{\delta^2}2]$, in which
\[H''(\frac 2{q^2}) \frac{\delta^2}2 = -\frac{q^4\log e}{4(q^2-2)}
\cdot \frac{1}{c^2q^4}
= -\frac{\log e}{4c^2 (q^2-2)}
< -\frac{1}{4c^2 q^2}.\]
Therefore, if we let 
$t(c, q) = \frac{1}{4c^2q^2\log q}$,
then the number of bad pairs is at most
\[q^2 \cdot O(n) \cdot 2 \cdot 2^{n \cdot [\log(q^2) + H''(2/q^2) \cdot \delta^2/2]}
< O(2^{n \cdot [\log(q^2) - t(c, q) \log q]})
= O(q^{(2 - t(c, q))n}).\]

On the other hand, we examine the number of monomials in $\hat g(x, y)$.
First we know $\hat g_{s, t}(x, y)$
is essentially an $n$-variate polynomial in $e_s(x[i]) \cdot e_t(y[i]), i\in [n]$ 
of degree at most $(h-l)$.
Since $e_s(x[i]) \cdot e_t(y[i]) \in \{0, 1\}$, we can further assume 
$\hat g_{s, t}(x, y)$ is multilinear
and thus contains at most
\[\sum_{i=0}^{h-l} \binom{n}{i} \le O(n) \cdot 2^{n \cdot H((h-l)/n)}
= O(n) \cdot q^{n\cdot H(2\delta) / \log q}\]
many monomials.
Consequently $\hat g(x, y) = \prod_{s, t\in [q]} \hat g_{s, t}(x, y) $ contains at most 
$O(n^{q^2}) \cdot q^{n \cdot H(2\delta) \cdot q^2 / \log q}$
many monomials.
Here
\begin{align*}
H(2\delta) \cdot \frac{q^2}{\log q} & <
2 \cdot 2\delta\log_2(\frac 1{2\delta}) \cdot \frac{q^2}{\log q}
= \frac{4}{cq^2} \log(\frac{cq^2}{2}) \cdot \frac{q^2}{\log q} \\
&= \frac 4c \cdot \frac{\log (\frac c2) + 2\log q}{\log q} 
\le  \frac 4c \cdot (2 + \log c)
\le \frac {8\log c}c.
\end{align*}
Hence the number of monomials in the expansion of $\hat g$
is bounded by $O(n^{q^2}) \cdot q^{n \cdot H(2\delta) \cdot q^2 / \log q}
= O(q^{n \cdot 8 \log c / c})$.

With the two quantities bounded,
we then can apply Lemma \ref{lem：polymethod} and derive a rigidity upper bound of $M_n$. 
For any $c > 4$,
\[\rig_{M_n} (O(q^{n \cdot 8 \log c/c})) = 
O(q^{n \cdot (2 - 1/4c^2q^2\log q)}).\]
We note that Theorem \ref{thm:rigGKron} immediately follows from this result. Indeed,
if changing variable $c$ as $c = \frac {128}{s} \log \frac 2{s}$
for any $0 < s < 1$, we have
\[\frac {8\log c}c
= \frac{s/16}{\log (2/s)} (\log(128/s) + \log\log(2/s)) 
\le \frac{s}8 \cdot \frac{\log(128/s)}{\log (2/s)}
= s \cdot \frac{\log(2^{7/8}/s^{1/8})}{\log (2/s)} < s.\]
Then rewriting the rigidity bound in terms of $s$ yields
\[\rig_{M_n}(q^{ns}) = O(q^{n \cdot (2-h(s, q))}),\]
where $h(s, q) = \frac{1}{4q^2\log q} \cdot \frac{(s/128)^2}{\log^2(2/s)}$.

\subsection{Linear circuits for General Kronecker power matrices}
\begin{theorem}
Suppose $q\ge 2$ is a positive integer, and $M \in \F^{q \times q}$ is a $q \times q$ matrix. 
For any positive integer $n$, letting $N  = q^n$, then $M_n = M^{\otimes n}$ has a depth-2 linear circuit of size 
$O(N^{1.5 - a_q})$, where $a_q = \Omega(\frac 1{q^2\log q})$ is a positive constant depending only on $q$.
\end{theorem}
\begin{proof}
We apply Corollary \ref{cor:rigLC} / Remark \ref{rem:asymLC} with the rigidity upper bound given in Theorem \ref{thm:rigGKron}. The base matrix is picked as 
$M^{\otimes k}$ for some positive integer $k$.
Then we know $M_n$ has a depth-2 linear circuit of size at most
$O((q^k \cdot q^{sk} + 
\sqrt{q^k \cdot q^{k \cdot (2 - h(s, q))}})^{n/k})
= O((q^{k \cdot (1 + s)} + 
q^{k \cdot (1.5 - h(s, q)/2)})^{n/k})$.
When $k$ is sufficiently large, this size can be further bounded by
$O((q^{k \cdot b_q})^{n/k}) = O(N^{b_q})$,
where $b_q = \max\{1+s, 1.5 - h(s, q)/2\}$.
By picking $s = 0.4$, we have $1 + s < 1.5 - h(s,q)/2$ for any $q \ge 2$.
Therefore, $b_q = 1.5 - h(0.4, q)/2 = 1.5 - \Omega(\frac 1{q^2\log q})$.
\end{proof}

\section{Disjointness matrices}
The smallest linear circuit for disjointness matrix $R_n$ previous to our work was given by Jukna and Sergeev in \cite{JS13}. Their method in a broad sense relies on combining 1-monochromatic blocks together along with the enlargement of $R_n$. To show more clearly the connection between their approach and ours, we first sketch their algorithm for generating the linear circuit. 

\begin{lemma}[\cite{JS13}, Lemma 4.2]
\label{lem:JS}
For positive integer $n$, let $N = 2^n$. Then disjointness matrix $R_n$ has a depth-2 linear circuit of size $O((1 + \sqrt{2})^n) = O(N^{1.272})$.
\end{lemma}
\begin{proof}[Proof sketch]
We partition the 1-entries in $R_n$ into (not necessarily uniformly sized) squares (i.e. $s \times s$ blocks for some $s$) and rectangles (i.e. $2s \times s$ blocks for some $s$) inductively. Let $s_n$ be the sum of side-lengths of the squares in the partition of $R_n$, and $r_n$ be the sum of shorter side-lengths of the rectangles.
We note that each square (resp. rectangle) $B$, if regarded as a rank-1 matrix,
can be rewritten as $B = U \times V$ with $\nnz(U) = s$ and $\nnz(V) = s$
(resp. $\nnz(U) = 2s$ and $\nnz(V) = s$).
Therefore, after being properly decomposed then combined, a partition of $R_n$ can be converted into a linear circuit of size at most
$2 \cdot (s_n + r_n)$.

Now we show how to obtain a partition of 1-entries in $R_n$ from the partition of $R_{n-1}$. From recursive definition,
\[
R_1 := \begin{pmatrix}
1 & 1 \\ 1 & 0
\end{pmatrix}, \quad
R_n := \begin{pmatrix}
R_{n-1} & R_{n-1} \\ R_{n-1} & 0
\end{pmatrix}.\]
One can see that if $B$ is a block in the partition of $R_{n-1}$, 
then $R_{n}$ contains exactly three identical copies of $B$.
Furthermore, for an $s \times s$ square in $R_{n-1}$, the three copies can be glued to form a $2s \times s$ rectangle and a $s \times s$ square. 
For an $2s \times s$ rectangle in $R_{n-1}$, the copies can form a
$2s \times 2s$ square and a $2s \times s$ rectangle in $R_n$.
Hence we have the recurrence
\[\begin{pmatrix} s_n \\ r_n \end{pmatrix}
= \begin{pmatrix} 1 & 2 \\ 1 & 1 \end{pmatrix} \times
\begin{pmatrix} s_{n-1} \\ r_{n-1} \end{pmatrix}.\]
In conjunction with $s_1 = r_1 = 1$ corresponding to the column-wise partition
of $R_1$, this recurrence has solution $s_n, r_n \le O((1 + \sqrt 2)^n)$,
which completes the proof.
\end{proof}

On the other hand, we observe that the row-wise partition of $R_1$ 
corresponds to a one-sided decomposition.
And further if Theorem \ref{thm:onesided} is applied with this decomposition, it produces a depth-2 linear circuit for
$R_n = R_1^{\otimes n}$ of the same size 
$O((\sqrt{1 \times 2} + \sqrt{1 \times 1})^n)
= O((1 + \sqrt{2})^n)$, as is given in Lemma \ref{lem:JS}.

This is in fact not a coincidence. Let
\[R_1 = \begin{pmatrix} 1 & 0 \\ 0 & 0 \end{pmatrix} \times
\begin{pmatrix} 1 & 1 \\ 0 & 0 \end{pmatrix} 
+ \begin{pmatrix} 0 & 0 \\ 0 & 1 \end{pmatrix} \times
\begin{pmatrix} 0 & 0 \\ 1 & 0 \end{pmatrix}
=: U_1 \times V_1 + U_2 \times V_2 \]
be the decomposition mentioned above.
By rewriting the procedure in Lemma \ref{lem:JS},
one can see an $2s \times s$ rectangle $B = U \times V$ 
in $R_{n-1}$ become precisely $B \otimes R_1 = (U \otimes U_1) \times (V \otimes V_1) + (U \otimes U_2) \times (V \otimes V_2)$ in $R_n$.
The first term is the glued $2s \times 2s$ square, while the second constitutes the 
remaining copy of $2s \times s$ rectangle.
Similarly, an $s \times s$ square $B = U \times V$ becomes
$B \otimes R_1 = (U \otimes V_1^T) \times (V \otimes U_1^T) + (U \otimes V_2^T) \times (V \otimes U_2^T)$ with terms respectively corresponding to the new $2s \times s$ rectangle and $s \times s$ square.

In other words, each step of the procedure is equivalent to a $\CU$ update in our approach, and as a result, the entire procedure performs exactly the same updates as does ours.
Therefore Jukna and Sergeev's method could be seen as a special case of our approach with the decomposition of $R_1$ mentioned above.

\subsection{Upper bound}
Inspired by Lemma \ref{lem:JS}, we would expect other partitions of small disjointness matrix to give smaller linear circuits. Consider the following partition for $R_3$.

{\begin{center}
\includegraphics[width=6cm]{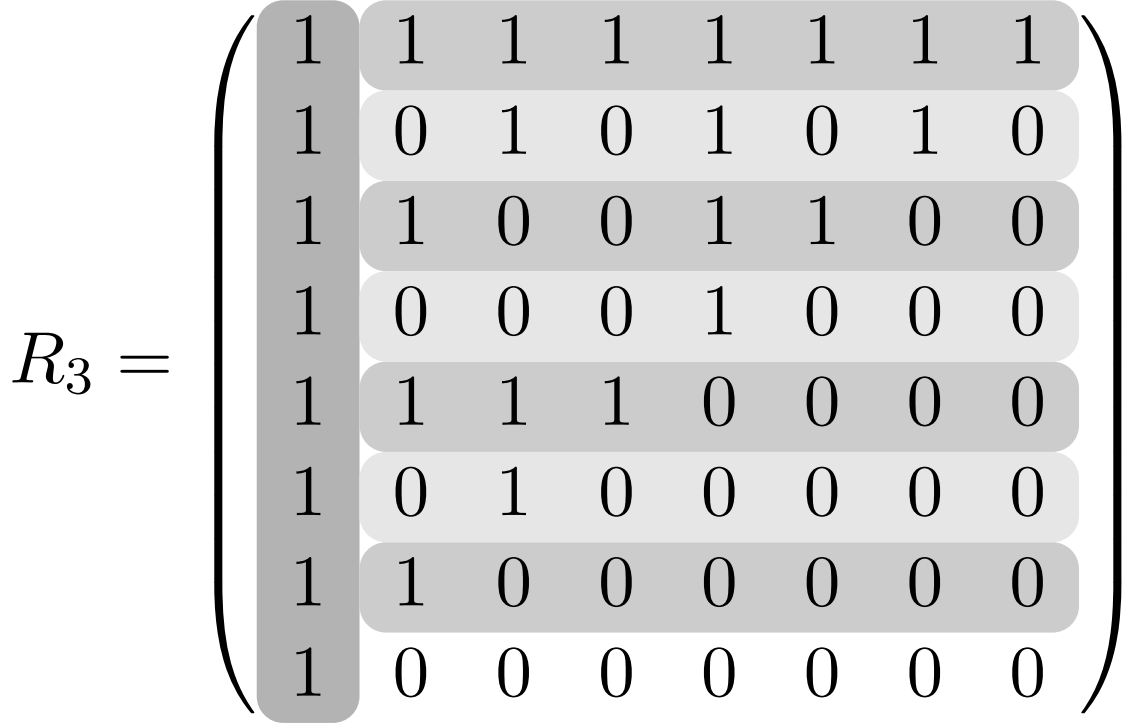}
\end{center}}

This partition corresponds to an 8-way decomposition $R_3 = \sum_{j = 1}^8 U_j \times V_j$. As this decomposition is not one-sided, we need to make use of the full power of Theorem \ref{thm:main}.
By calculation, $\alpha_1 \approx 2.616, \beta \approx 0.0827 > \alpha_2 - \alpha_1 \approx 0.0724$. 
With this condition satisfied, we can apply Theorem \ref{thm:main} on $R_3$ and thus conclude
that $R_3^{\otimes n}$ has a depth-2 linear circuit of size $\exp(\alpha_1)^{n+o(n)} = 
O(2^{1.258\cdot 3n})$. After rewriting in terms of $N$, we have the following theorem.

\begin{theorem}
For any positive integer $n$, let $N = 2^n$. Then
$R_n$ has a depth-2 linear circuit of size $O(N^{1.258})$.
\end{theorem}

\section*{Acknowledgements}

We would like to thank the anonymous reviewers for their helpful comments. This research was supported in part by a grant from the Simons Foundation (Grant Number 825870 JA).

\addcontentsline{toc}{section}{References}
\bibliographystyle{alpha}
\bibliography{ref}

\begin{thebibliography}{BGKM22}

\bibitem[Alm21]{Alm21}
Josh Alman.
\newblock Kronecker products, low-depth circuits, and matrix rigidity.
\newblock In {\em Proceedings of the 53rd Annual ACM SIGACT Symposium on Theory
  of Computing}, STOC 2021, page 772–785, New York, NY, USA, 2021.
  Association for Computing Machinery.

\bibitem[AW15]{AW15}
Josh Alman and Ryan Williams.
\newblock Probabilistic polynomials and hamming nearest neighbors.
\newblock In {\em 2015 IEEE 56th Annual Symposium on Foundations of Computer
  Science}, pages 136--150, 2015.

\bibitem[AW17]{AW17}
Josh Alman and Ryan Williams.
\newblock Probabilistic rank and matrix rigidity.
\newblock In {\em Proceedings of the 49th Annual ACM SIGACT Symposium on Theory
  of Computing}, STOC 2017, page 641–652, New York, NY, USA, 2017.
  Association for Computing Machinery.

\bibitem[AWY14]{abboud2014more}
Amir Abboud, Ryan Williams, and Huacheng Yu.
\newblock More applications of the polynomial method to algorithm design.
\newblock In {\em Proceedings of the twenty-sixth annual ACM-SIAM symposium on
  Discrete algorithms}, pages 218--230. SIAM, 2014.

\bibitem[BCS13]{burgisser2013algebraic}
Peter B{\"u}rgisser, Michael Clausen, and Mohammad~A Shokrollahi.
\newblock {\em Algebraic complexity theory}, volume 315.
\newblock Springer Science \& Business Media, 2013.

\bibitem[Bea84]{beauchamp1984applications}
Kenneth~George Beauchamp.
\newblock {\em Applications of Walsh and related functions, with an
  introduction to sequency theory}, volume~2.
\newblock Academic press, 1984.

\bibitem[BGKM22]{bhargava2022fast}
Vishwas Bhargava, Sumanta Ghosh, Mrinal Kumar, and Chandra~Kanta Mohapatra.
\newblock Fast, algebraic multivariate multipoint evaluation in small
  characteristic and applications.
\newblock In {\em Proceedings of the 54th Annual ACM SIGACT Symposium on Theory
  of Computing}, pages 403--415, 2022.

\bibitem[BK21]{BK21}
L\'{a}szl\'{o} Babai and Bohdan Kivva.
\newblock {Matrix Rigidity Depends on the Target Field}.
\newblock In {\em 36th Computational Complexity Conference (CCC 2021)}, volume
  200, pages 41:1--41:26, 2021.

\bibitem[BL04]{burgisser2004lower}
Peter B{\"u}rgisser and Martin Lotz.
\newblock Lower bounds on the bounded coefficient complexity of bilinear maps.
\newblock {\em Journal of the ACM (JACM)}, 51(3):464--482, 2004.

\bibitem[Cha94]{chazelle1994spectral}
Bernard Chazelle.
\newblock A spectral approach to lower bounds.
\newblock In {\em Proceedings 35th Annual Symposium on Foundations of Computer
  Science}, pages 674--682. IEEE, 1994.

\bibitem[DE19]{dvir2019matrix}
Zeev Dvir and Benjamin~L Edelman.
\newblock Matrix rigidity and the croot-lev-pach lemma.
\newblock {\em Theory Of Computing}, 15(8):1--7, 2019.

\bibitem[DL20]{DL19}
Zeev Dvir and Allen Liu.
\newblock Fourier and circulant matrices are not rigid.
\newblock {\em Theory OF Computing}, 16(20):1--48, 2020.

\bibitem[Fri93]{DBLP:Fri93}
Joel Friedman.
\newblock A note on matrix rigidity.
\newblock {\em Combinatorica}, 13(2):235--239, 1993.

\bibitem[GRZM14]{multisect}
Izrail~Solomonovich Gradshteyn, I~M Ryzhik, Daniel Zwillinger, and Victor Moll.
\newblock {\em {Table of integrals, series, and products; 8th ed.}}
\newblock Academic Press, Amsterdam, Sep 2014.

\bibitem[JS13]{JS13}
Stasys Jukna and Igor Sergeev.
\newblock Complexity of linear boolean operators.
\newblock {\em Foundations and Trends® in Theoretical Computer Science},
  9(1):1--123, 2013.

\bibitem[Kiv21]{kivva2021improved}
Bohdan Kivva.
\newblock Improved upper bounds for the rigidity of kronecker products.
\newblock {\em arXiv preprint arXiv:2103.05631}, 2021.

\bibitem[Lok00]{Lok00}
Satyanarayana~V. Lokam.
\newblock On the rigidity of vandermonde matrices.
\newblock {\em Theoretical Computer Science}, 237(1):477--483, 2000.

\bibitem[Lok01]{lokam2001spectral}
Satyanarayana~V Lokam.
\newblock Spectral methods for matrix rigidity with applications to size--depth
  trade-offs and communication complexity.
\newblock {\em Journal of Computer and System Sciences}, 63(3):449--473, 2001.

\bibitem[Lok09]{lokam2009complexity}
Satyanarayana~V Lokam.
\newblock Complexity lower bounds using linear algebra.
\newblock {\em Foundations and Trends{\textregistered} in Theoretical Computer
  Science}, 4(1--2):1--155, 2009.

\bibitem[Mid05]{midrijanis}
Gatis Midrijanis.
\newblock Three lines proof of the lower bound for the matrix rigidity.
\newblock {\em CoRR}, abs/cs/0506081, 2005.

\bibitem[Mor73]{morgenstern1973note}
Jacques Morgenstern.
\newblock Note on a lower bound on the linear complexity of the fast fourier
  transform.
\newblock {\em Journal of the ACM (JACM)}, 20(2):305--306, 1973.

\bibitem[MS77]{macwilliams1977theory}
Florence~Jessie MacWilliams and Neil James~Alexander Sloane.
\newblock {\em The theory of error correcting codes}, volume~16.
\newblock Elsevier, 1977.

\bibitem[NW96]{nisan1996lower}
Noam Nisan and Avi Wigderson.
\newblock Lower bounds on arithmetic circuits via partial derivatives.
\newblock {\em Computational complexity}, 6(3):217--234, 1996.

\bibitem[Pud00]{pudlak2000note}
Pavel Pudl{\'a}k.
\newblock A note on the use of determinant for proving lower bounds on the size
  of linear circuits.
\newblock {\em Information processing letters}, 74(5-6):197--201, 2000.

\bibitem[Raz02]{raz2002complexity}
Ran Raz.
\newblock On the complexity of matrix product.
\newblock In {\em Proceedings of the thiry-fourth annual ACM symposium on
  Theory of computing}, pages 144--151, 2002.

\bibitem[SSS97]{SSS1997riglb}
M.A. Shokrollahi, D.A. Spielman, and V.~Stemann.
\newblock A remark on matrix rigidity.
\newblock {\em Information Processing Letters}, 64(6):283--285, 1997.

\bibitem[Val77]{valiant1977graph}
Leslie~G Valiant.
\newblock Graph-theoretic arguments in low-level complexity.
\newblock In {\em International Symposium on Mathematical Foundations of
  Computer Science}, pages 162--176. Springer, 1977.

\bibitem[Won22]{wong2022walsh}
Hiu~Yung Wong.
\newblock Walsh--hadamard gate and its properties.
\newblock In {\em Introduction to Quantum Computing}, pages 153--162. Springer,
  2022.

\end{thebibliography}

\end{document}